\documentclass[11pt,letterpaper]{article}
\usepackage{graphicx} % Required for inserting images
\usepackage[margin=1in]{geometry}
\usepackage{booktabs} % For formal tables
\usepackage[ruled]{algorithm2e} % For algorithms

\usepackage{fluff}
\theoremstyle{definition}
\newtheorem{MainResult}{Main Result}

\usepackage{mathtools}
\usepackage{subfig}
\usepackage{enumitem}
\usepackage{thm-restate}
\usepackage{hyperref}

\title{Stronger core results with multidimensional prices}

\author{
    Mark Braverman\thanks{Princeton University, \texttt{mbraverm@gmail.com}}
    \and
    Jingyi Liu\thanks{Princeton University, \texttt{jingyi.liu@princeton.edu}}
    \and
    Eric Xue\thanks{Princeton University, \texttt{ex3782@princeton.edu}}
    \and
    Chenghan Zhou\thanks{Stanford University, \texttt{chzhou@stanford.edu}}
}

\begin{document}

\maketitle

\begin{abstract}
We study one-sided matchings with endowments in the absence of money.
It is well-known that a competitive equilibrium may not always exist and that the strong core may be empty in this setting \cite{HyllandZ1979}. 
We propose a generalization of competitive equilibria that associates each item with a multi-dimensional price.
We show that this solution concept always exists and resides within the rejective core \cite{Konovalov2005}.
Rejective core stability is strictly stronger than weak core stability: allocations in the rejective core are elements of the weak core, but the opposite is not true.
Moreover, we show that the rejective core always converges to the set of competitive equilibria with multi-dimensional prices as the economy grows, demonstrating core convergence in a setting without non-satiation.
\end{abstract}

\addtocounter{page}{-1}

\newpage

% Paper body
\section{Introduction}

One-sided matching without money is a fundamental problem in economics with applications ranging from course allocation to school choice.
At a high level, the problem consists of $n$ agents with preferences over $n$ (indivisible) goods, and the goal is to assign at most one good to each agent in an efficient and fair manner.
The seminal work of \cite{HyllandZ1979} proposes a pseudomarket approach to this problem when agents possess von Neumann-Morgenstern (vNM) utility functions.
The approach endows the agents with equal exogenous budgets of artificial money and asks each agent to purchase her favorite lottery over the goods subject to equilibrium prices and her budget.
\cite{HyllandZ1979} prove that a competitive equilibrium always exists and is Pareto efficient and envy-free.
Their  application of the concept of competitive equilibria from equal incomes (CEEI) \cite{Varian1974} to the allocation of indivisible goods inspired a long line of work \cite{BudishCKM2013, HeMPY2018, GulP2020, EcheniqueMZ2021Constrained, GulPZ2024, NguyenT2024} that sought to extend their result beyond the setting of one-sided allocation with unit-demand agents.\footnote{Here and henceforth, by unit-demand, we specifically mean agents with vNM utilities.}

A common thread between these works, and a fundamental feature of CEEI, is that agents (that impose the same feasibility constraints) have equal opportunity: they face the same prices, and they possess the same income.
This equity is desirable when the goods being allocated are collectively owned.
However, when the goods in the market are privately owned, bestowing an agent with the same opportunity as her peers may be undesirable if she brings comparatively more to the table. 
To extend the result of \cite{HyllandZ1979} to such markets, we study the problem of one-sided matching with endowments and agents with vNM utility functions.
For concision, we will drop the qualifier that the agents have vNM utility functions moving forward, but when we refer to one-sided matchings with endowments, this is what we mean.
The setting is nearly identical to that of \cite{HyllandZ1979} but assumes that each agent possesses an initial endowment which she contributes to the market.
A competitive equilibrium for this setting consists of a randomized matching and a price for each item such that each agent receives her favorite lottery subject to these prices and a budget that is endogenously determined by her initial endowment.

The study of one-sided matchings with endowments in the absence of money dates back to \cite{ShapleyS1974}, who examined the problem assuming integral endowments.
They show that the weak core\footnote{An allocation resides in the weak core if no coalition can trade within itself to make each member better off.} is always non-empty but may include non-competitive allocations.
\cite{RothP1977} observe that the allocation considered by \cite{ShapleyS1974} is not competitive because there exists a coalition of agents who can trade among themselves to make each member better off upon receiving their allocations.
They define the notion of stability to exclude allocations of this kind and appeal to a result of \cite{ShapleyS1974} to show that a stable allocation always exists.
When agents possess strict preferences, \cite{RothP1977} prove that a competitive equilibrium always exists and is unique.
The equilibrium allocation is strong core stable\footnote{An allocation resides in the strong core if no coalition can trade within itself to make some member better off without making the other members worse off. Note that under strict preferences, the weak and strong core coincide.} and can be efficiently and truthfully computed by the Top Trading Cycle algorithm \cite{ShapleyS1974, Roth1982}.

Unfortunately, it is well-known that in general, a competitive equilibrium may not exist and that the strong core may be empty. 
\cite{HyllandZ1979} provide a counterexample and observe that agents may wish to purchase bundles that are cheaper than their endowments due to satiability.
A line of work \cite{bergstrom1976, DrezeM1980, Makarov1981, AumannD1986, Mas-Colell1992, PolemarchakisS1993, Le2017, McLennan2018, EcheniqueMZ23, GargTV2024} seeks to circumvent this impossibility through various means (and in more general contexts).
\cite{bergstrom1976, PolemarchakisS1993, Le2017} give sufficient conditions for the existence of a competitive equilibrium.
\cite{Mas-Colell1992, Le2017, EcheniqueMZ23, GargTV2024} define notions of competitive equilibria with slack and prove that such equilibria always exist.
\cite{GargTV2024} additionally shows that competitive equilibria with slack lie in the approximate core.
\cite{HyllandZ1979} loosely conjecture that a redistribution of the surplus from satiated agents may resolve the problem of satiability.
In line with this intuition, \cite{DrezeM1980, Makarov1981, AumannD1986} propose a solution concept called \textit{dividend equilibrium} that allows for this redistribution and gives sufficient conditions for existence.
\cite{McLennan2018} extends this concept to economies with production.

Despite the tremendous progress on the problem of one-sided matching with endowments, a solution concept that always exists and exhibits the desirable properties of competitive equilibria remains elusive.
Our first main result is an equilibrium concept, which we term lexicographic dividend equilibria (LDE), that always exists and lies in the \textit{rejective core} \cite{Konovalov2005}.
An allocation resides in the rejective core if no coalition of agents, some endowed with their allocations and others with their initial endowments, can trade within itself such that those in the former group weakly improve while those in the latter group strictly improve (and at least one member of the coalition strictly improves).\footnote{\cite{Konovalov2005} defines the rejective core slightly differently. Their definition requires the blocking allocation to also be well-defined for agents outside of the coalition, and for these agents to weakly prefer the blocking allocation over their initial endowments. They work with a continuum of agents, so this additional requirement is without loss of generality. In finite economies, it unnecessarily complicates analysis but can be handled for our results (see \cite{MurakamiU2017}).}
Note that membership in the rejective core implies \textit{ex ante} Pareto optimality, individual rationality, and stability in the sense of \cite{RothP1977}.
We show that the opposite is not true.

\begin{MainResult}[Informal, see Theorems~\ref{theorem:LDE-existence} and~\ref{theorem:rejective-core-contains-LDEs}]\label{MainResult:existence}
For all one-sided matching markets with endowments, a lexicographic dividend equilibrium (LDE) exists and resides in the rejective core.
\end{MainResult}

Unlike a competitive equilibrium, which consists of an allocation and a price \textit{vector}, an LDE consists of three components: an allocation, a price \textit{matrix} $P \in \RR^{d \times n}$, and a redistribution matrix $A \in \RR^{d \times n}_+$ where $\RR_+$ denotes the non-negative reals.
Each column $p_j \in \RR^d$ of the price matrix corresponds to a commodity $j$, and each column $\alpha_i \in \RR^d_+$ of the redistribution matrix corresponds to the surplus redistributed to agent $i$.
We require that the first non-zero entry of $p_j$ be positive, while subsequent entries are allowed to be any real number.
The budget set of an agent with initial endowment $\omega \in \RR^n_+$ and redistributed surplus $\alpha \in \RR^n_+$ is the \textit{closure} of $\{x \in \Delta^n : P x \leq_{\mathrm{lex}} P \omega + \alpha\}$ where $\Delta^n$ denotes the $(n-1)$-simplex and $\leq_{\mathrm{lex}}$ denotes the lexicographic ordering on $\RR^d$.\footnote{$v \leq_{\mathrm{lex}} w$ if either (1) there exists $k$ such that $v_{\ell} = w_{\ell}$ for all $\ell < k$ and $v_{k} < w_{k}$ or (2) $v_{k} = w_{k}$ for all $k < d$ and $v_d \leq w_d$.}
In the same vein, a bundle $x$ is cheaper than bundle $y$ at prices $P$ if $P x <_{\mathrm{lex}} P y$.
We now have sufficient overhead to define a LDE: an allocation, a price matrix, and a redistribution matrix constitute a LDE if each agent receives her favorite affordable bundle, and there is no cheaper bundle that yields the same utility.
When $d = 1$, the set of LDEs coincides with the set of dividend equilibria \cite{AumannD1986}.

Our proof of Main Result~\ref{MainResult:existence} suggests a natural interpretation for the lexicographic nature of LDEs.
The rows of the price and redistribution matrices correspond to currencies, and each currency is infinitely more valuable than those that follow.
As a consequence, commodities with zero price in the first currency are free for agents with positive endowments of the first currency, even though these commodities may possess positive prices in later currencies.
This is because these agents can always convert an $\varepsilon$ amount of the first currency into an infinite amount of the second, so the closure of $\{x \in \Delta^n : P x \leq_{\mathrm{lex}} P \omega + \alpha\}$ treats goods with zero price in the first currency as costless.
On the other hand, commodities with positive price in the first currency are unobtainable for agents with no endowment of the first currency, even though these agents may possess a surplus of later currencies.
This surplus instead is redistributed.

A loosely related solution concept is that of quasi-equilibrium \cite{Debreu1962}.
An allocation and a price (vector) constitute a quasi-equilibrium if agents with positive income receive their favorite affordable bundles, while those with no income are only required to receive a bundle that is free.
\cite{Debreu1962} introduced this concept specifically due to the ``basic mathematical difficulty that the demand correspondence of a consumer may not be upper semi-continuous when his wealth equals the minimum compatible with his consumption set.''
Prior work circumvented this difficulty through various means.
The assumption that each agent possesses a strictly positive endowment dates back as far as the seminal work of \cite{ArrowD1954} and also appears in \cite{Le2017, McLennan2018}.
The allowance of slack in income \cite{Mas-Colell1992, EcheniqueMZ23, GargTV2024} immediately ensures that no agent lacks wealth.
Separately, \cite{ArrowD1954} assume the existence of a universally desired commodity that is owned in positive quantity by each agent\footnote{As \cite{ArrowD1954} work in an economy with production, the actual condition they give is the existence of a universally desired commodity and a universally productive commodity that is owned in positive quantity by each agent. In a pure exchange economy like ours, the latter condition means positive ownership of a universally desired good.}, and a line of work \cite{McKenzie1959, ArrowH1971, Moore1975, Eaves1985, Maxfield1997} sought to weaken this assumption as much as possible. 
A common thread in this line of work was the idea of ``irreducibility'': it should not be possible to partition the market into two groups such that one group is unable to supply any commodities that the other group wants yet wants commodities that the other group has.
The concept of quasi-equilibrium prevents those that want but cannot supply from consuming these goods by simply endowing them with no income.

The mathematically difficulty highlighted by \cite{Debreu1962} remains an obstacle for us.
The insight of our solution concept is essentially that a competitive equilibrium exists in the sub-market of agents with no income.
It turns out we are not the first to have this insight, although we did not become aware of this fact until later.
\cite{AbrahamSW1996} show that while competitive equilibria may not exist, there exists a sequence of prices with vanishing excess demand that approaches a quasi-equilibria.
We develop this insight into a solution concept that always exists and demonstrate its properties.
It should not come as a surprise that the proof of Main Result~\ref{MainResult:existence} appeals to this sequence of prices.
At a high level, we assign a lexicographic price to a commodity with vanishing price based on the rate at which it vanishes.
The number of distinct convergence rates determines the number of currencies in the LDE.

\subsection{Core Convergence}

It is well-known that the concept of the core \cite{VonNeumannM1944, Gillies1959, Aumann1961, Peleg1963} has a deep connection with the concept of competitive equilibrium in pure exchange economies.
\cite{Scarf1967} gives a sufficient condition for the weak core to be non-empty for a general $n$-person game without side payments, but it has been known since \cite{DebreuS1963} that under certain assumptions, the weak core of a pure exchange economy is always non-empty as it contains the set of competitive equilibria.
These assumptions include insatiability, strict convexity, and strictly positive endowments.
\cite{DebreuS1963} in fact show something even stronger under these assumptions: the weak core converges to the set of competitive equilibria as the economy is replicated.\footnote{A replica economy is one that contains the same number of duplicates of each original agent.}

The idea of core convergence dates back to \cite{Edgeworth1881} and gives weight to the idea of competitive equilibrium as a solution concept: competitive equilibria are the only outcomes that cannot be blocked by a coalition of agents in large markets.
\cite{Aumann1964} verified this thesis in a different way by demonstrating that the core coincides with the set of competitive equilibria when there is a continuum of agents.
Interestingly, the result of \cite{Aumann1964} does not require convexity but insatiability remains indispensable.
The generality of this result compared to that of \cite{DebreuS1963} motivated a long line of work to address this gap.
We refer the reader to \cite{Anderson1992} for an excellent survey.

% \cite{Conley1994} studies core convergence with satiable agents but in economies with public goods.
% He gives sufficient conditions for the core of replica economies to converge to the set of Lindahl equilibria \cite{Lindahl1958}.
% An important difference between economies with public goods and those without is that satiability in public goods helps convergence because it prevents the grand coalition from taking advantage of its cost advantage over sub-coalitions.
% In contrast, agents satiated in private goods are often members of weak blocking coalitions because they are willing to distribute their surplus to other members.

There is a more recent line of work \cite{Konovalov2005, MiyazakiT2012, MurakamiU2017} on core convergence in pure exchange economies with satiable agents.
\cite{Konovalov2005, MiyazakiT2012} give sufficient conditions for identifying the rejective core with the set of dividend equilibria \cite{AumannD1986} when there is a continuum of agents.
\cite{MurakamiU2017} give sufficient conditions for this characterization to hold for replica economies.
In general, a one-sided matching market with endowments may not satisfy any of the sufficient conditions considered by these works.
Our second main result is that the limit of the rejective core as the economy replicates \textit{always} coincides with the set of LDEs.

\begin{MainResult}[Informal, see Theorem~\ref{theorem:core-convergence}]
For all one-sided matching markets with endowments, the intersection over all replica economies of the rejective core coincides with the set of LDEs.
\end{MainResult}

Like much of the prior work, the proof makes use of ideas from \cite{DebreuS1963}.
We find a hyperplane that separates the set of preferred trades from the negative orthant, remove any positively priced items, and recurse.
Some care is needed to ensure that the recursion terminates.

\section{Preliminaries}

The problem of one-sided matching with endowments consists of $n$ agents and $n$ goods.
Each agent $i$ possesses an endowment $\omega_{ij} \in \RR_+$ of each good $j$, as well as a utility $u_{ij} \in \RR_+$.
Here, $\RR_+$ denotes the non-negative reals.
We assume $\sum_i \omega_{ij} = 1$ for all goods $j$, that is, all goods are owned by the agents.
Note that we do not assume that $\omega_i \gg 0$ nor even that $\omega_i \not= 0$.
We use $\omega_i \coloneqq (\omega_{ij})_j$ and $u_i \coloneqq (u_{ij})_j$ to denote agent $i$'s endowment and utilities, respectively, while $\omega \coloneqq (\omega_i)_i$ and $u \coloneqq (u_i)_i$ will denote the endowment and utility profiles of the economy, respectively.

We assume agents possess von Neumann-Morgenstern utility functions, so $(u, \omega)$ suffices to completely identify an economy.
Let $\Delta^n \coloneqq \{x \in \RR^n_+ : \sum_j x_j = 1\}$ and $\Delta^n_- \coloneqq \{x \in \RR^n_+ : \sum_j x_j \leq 1\}$ denote the $(n-1)$-dimensional simplex and its lower contour set (in the non-negative orthant), respectively.
Agent $i$'s utility for the lottery $x_i \in \Delta^n_-$ is then given by $u_i \cdot x_i$.
Agents with utility functions of this kind are also known as unit-demand.
We sometimes write $x_i \succeq_i y_i$ to indicate that $u_i \cdot x_i \geq u_i \cdot y_i$ and $x_i \succ_i y_i$ when the inequality is strict.
We say an agent is satiated at $x_i$ if there does not exist $y_i \in \Delta^n_-$ such that $y_i \succ_i x_i$.

The set of allocations is given by $\mathcal{M} \coloneqq \{x \in \RR^{n \times n}_+ : x_i \in \Delta^n \,\forall\, i \wedge x_j \in \Delta^n \,\forall\, j\}$.
An allocation is integral if $x_{ij} \in \{0, 1\}$ for all agents $i$ and goods $j$.
By the Birkhoff-von Neumann theorem, we can decompose any element of $\mathcal{M}$ into a distribution over integral allocations.

A competitive equilibrium is desirable in part because it is \textit{ex ante}, or fractionally, Pareto optimal, individually rational, and strong core stable.
An allocation $x \in \mathcal{M}$ is 
\begin{itemize}
    \item Fractionally Pareto optimal (fPO) if there does not exist $y \in \mathcal{M}$ such that $y_i \succeq_i x_i$ for all agents $i$ with strict preference for some agent.
    \item Individually rational (IR) if $x_i \succeq_i \omega_i$ for all agents $i$.
    \item Strong core stable if there does not exist a coalition $C$ of agents and consumptions $y \in \Delta^{n \times C}_-$ such that $\sum_{i \in C} (y_i - \omega_i) \leq 0$ and $y_i \succeq_i x_i$ for all  $i \in C$ with strict preference for some agent.
\end{itemize}
In other words, an allocation is strong core stable if no coalition is able to redistribute the goods it collectively owns to make some member better off without making the others worse off.
A coalition capable of this feat is said to weakly block the allocation.

However, it is known that in economies with satiation, a competitive equilibrium may not exist and the strong core may be empty because satiated agents are willing to redistribute their surplus and thus often find themselves members of weakly blocking coalitions.
The literature has proposed various alternative notions of stability to address the presence of these agents.

The notion of weak core stability simply excludes these agents from coalition membership by strengthening what it means to block an allocation.
A coalition $C$ strongly blocks an allocation $x$ if there exists consumptions $y \in \Delta^{n \times C}_-$ such that $\sum_{i \in C} (y_i - \omega_i) \leq 0$ and $y_i \succ_i x_i$ for all $i \in C$.
Note that a strong blocking coalition cannot contain satiated agents as there are no strictly preferred consumptions.
An allocation is weak core stable if it has no strong blocking coalitions.
A weak core allocation is necessarily individually rational.

In general, the weak core can be quite large \cite{AumannD1986}, and allocations in the weak core need not be efficient.
\cite{RothP1977} realized that this inefficiency stems from the fact that although no coalition can redistribute its collective \textit{endowment} to strictly improve, a coalition may wish to redistribute its collective \textit{allocation} to Pareto improve.
An allocation $x$ is stable if it is weak core stable with respect to the economy $(u, \omega)$ and strong core stable with respect to the economy $(u, x)$.\footnote{The definition given by \cite{RothP1977} only requires weak core stability with respect to the economy $(u, x)$ as stable allocations, as defined here, may not exist if allocations must be integral. If we dispense with this requirement, then one can hope for something stronger.}
Stable allocations are \textit{ex ante} Pareto optimal and weak core stable.

To characterize the set of dividend equilibrium \cite{AumannD1986}, \cite{Konovalov2005} introduces the notion of rejective core stability.

\begin{definition}[Rejective core]
A coalition $C$ rejects an allocation $x$ if there exist consumptions $y \in \Delta^{n \times C}_-$ and a partition $(C_1, C_2)$ of $C$ such that 
\begin{itemize}
    \item $\sum_{i \in C} y_i \leq \sum_{i \in C_1} \omega_i + \sum_{i \in C_2} x_i$
    \item $y_i \succ_i x_i$ for all $i \in C_1$
    \item $y_i \succeq_i x_i$ for all $i \in C_2$ with strict preference for some agent if $C_1 = \varnothing$.
\end{itemize}
An allocation is rejective core stable if no coalition rejects it.
\end{definition}

Intuitively, for any coalition that rejects an allocation $x$, there exists a sufficiently large (replica) economy in which it can realize the desired consumption $y$ without infringing on the individual rationality of the agents outside of the coalition (e.g., by giving some of them what they would have gotten in $x$ while leaving others with their endowments).\footnote{Technically, it may be impossible to realize the consumption $y$ without violating the individual rationality of one replica of each agent, but we digress as the purpose of the discussion is to provide intuition.}
An allocation is rejective core stable if no coalition rejects it.
The set of stable allocations contains the rejective core, but opposite is false.

\begin{table}
    \centering
    \subfloat[Utilities]{\begin{tabular}{c|cccc} 
        & $A$ & $B$ & $C$ & $D$ \\
        \hline
        1 & 0 & 0 & 1 & 1 \\
        2 & 2 & 0 & 1 & 0 \\
        3 & 0 & 1 & 2 & 0 \\
        4 & 0 & 0 & 0 & 1 \\
    \end{tabular}}
    \quad\quad\quad\quad
    \subfloat[Endowments]{\begin{tabular}{c|cccc} 
        & $A$ & $B$ & $C$ & $D$ \\
        \hline
        1 & $1/2$ & 0 & 0 & $1/2$\\
        2 & 0 & 0 & 1 & 0 \\
        3 & 0 & 0 & $1/2$ & $1/2$ \\
        4 & 0 & $1/2$ & 0 & $1/2$
    \end{tabular}}
    \quad\quad\quad\quad
    \subfloat[Allocation]{\begin{tabular}{c|cccc} 
        & $A$ & $B$ & $C$ & $D$ \\
        \hline
        1 & 0 & 0 & $1/2$ & $1/2$ \\
        2 & $1/2$ & 0 & $1/2$ & 0 \\
        3 & 0 & $1/2$ & $1/2$ & 0 \\
        4 & 0 & 0 & 0 & 1
    \end{tabular}}
    \caption{A stable allocation that is not rejective core stable}
    \label{table:stable-not-rejective-core-example}
\end{table}

We give an example demonstrating this fact in Table~\ref{table:stable-not-rejective-core-example}.
The example consists of agents $\{1,2,3,4\}$ and goods $\{A, B, C, D\}$.
Each of $A$ and $B$ exist in a $1/2$ share, while each of $C$ and $D$ exist in a $3/2$ share.
The economy can be made integral by duplicating the market if one wishes.
The allocation given in the example clearly resides in the weak core with respect to the economy $(u, \omega)$: all but agent 3 cannot be made happier, so a strong blocking coalition cannot include them, and agent 3 cannot strictly improve her standing by her lonesome.
It is also strong core with respect to the economy $(u, x)$ since to make agent 3 happier, either agent 1 or agent 2 needs to give up some of their shares of $C$.
However, there are no shares of $A$ left to compensate agent 2 with, as we already allocated all shares of $A$ to her, and to compensate agent 3, agent 4 would need to relinquish some of her shares of $D$, which would make her strictly worse off.
However, note that the coalition consisting of agent 1 with her allocation and agent 3 with her endowment rejects the allocation.

\section{Lexicographic Dividend Equilibrium}

Two obstacles prevent the existence of a competitive equilibrium in economies with satiation.
The first is that a competitive equilibrium requires no surplus but a satiated agent may spend strictly less than the value of her endowment.
\cite{HyllandZ1979} give the example in Table~\ref{table:HZ-example} to illustrate this point.
The example consists of agents 1, 2, and 3 and goods $A$ and $B$.
Agents 1 and 2 prefer good $A$, while agent 3 prefers good $B$.
Each agent is endowed with a $1/3$ share of $A$ and a $2/3$ share of $B$, so there is one unit of $A$ and two units of $B$ in the economy.
We sketch a proof of the non-existence of a competitive equilibrium (see \cite{EcheniqueMZ23} for a formal proof).

\begin{table}
    \centering
    \subfloat[Utilities]{\begin{tabular}{c|cc} 
        & $A$ & $B$ \\
        \hline
        1 & 2 & 1 \\
        2 & 2 & 1 \\
        3 & 0 & 1
    \end{tabular}}
    \quad\quad\quad\quad
    \subfloat[Endowments]{\begin{tabular}{c|cc} 
        & $A$ & $B$ \\
        \hline
        1 & $1/3$ & $2/3$ \\
        2 & $1/3$ & $2/3$ \\
        3 & $1/3$ & $2/3$
    \end{tabular}}
    \quad\quad\quad\quad
    \subfloat[Allocation]{\begin{tabular}{c|cc} 
        & $A$ & $B$ \\
        \hline
        1 & $1/2$ & $1/2$ \\
        2 & $1/2$ & $1/2$ \\
        3 & 0 & 1
    \end{tabular}}
    \caption{Example from \cite{HyllandZ1979}}
    \label{table:HZ-example}
\end{table}

A competitive equilibrium must allocate one unit of $B$ to agent 3 as anything else is inefficient.
Moreover, since agents 1 and 2 have the same endowment and hence equal income, they each must receive a $1/2$ share of $A$ and a $1/2$ share of $B$ in a competitive equilibrium.
If the price of $A$ exceeds the price of $B$, then agents 1 and 2 cannot afford their bundles.
On the other hand, if the price of $B$ exceeds the price of $A$, then agent 3 cannot afford her bundle.
But if the price of the two commodities coincide, then agents 1 and 2 would each attempt to consume an entire unit of $A$.

As noted by \cite{HyllandZ1979}, the allocation \textit{should} be $(1/2, 1/2)$ for agents 1 and 2 and $(0,1)$ for agent 3, and the market clearing prices \textit{should} be $(2,0)$.
However, as mentioned earlier, the issue with this assignment of prices is that agents 1 and 2 lack the income to purchase the lotteries they should receive.
Meanwhile, agent 3 squats on $2/3$ units of surplus.
The dividend equilibrium of \cite{AumannD1986} addresses this issue by redistributing this surplus to agents 1 and 2: by endowing each of them with an additional $1/3$ unit of income, they can now afford their allocations.
We give a definition of dividend equilibrium specialized to our setting.

\begin{definition}[Dividend equilibrium \cite{AumannD1986}]
A dividend equilibrium of the economy $(u, \omega)$ consists of an allocation $x \in \mathcal{M}$, prices $p \in \RR^n_+$, and dividends $\alpha \in \RR^n_+$ such that $y_i \succ_i x_i$ implies that $p \cdot y > p \cdot \omega_i + \alpha_i$ for all agents $i$.
\end{definition}

There is no guarantee that a dividend equilibrium exists.
Consider the example in Table~\ref{table:our-example}. 
Individual rationality requires that agent 1 receives at least a $1/2$ share of $A$ and that agent 2 receives at least a $1/2 - O(\delta)$ share.
These constraints preclude the existence of a satiated agent at an equilibrium, so there is no surplus to redistribute, and we return to looking for a competitive equilibrium.
In a competitive equilibrium, the price of $A$ must exceed the price of $B$ as otherwise, agents 1 and 2 would each attempt to consume an entire unit of $A$.
Similarly, the price of $B$ must exceed the price of $C$ as otherwise, agent 3 would attempt to consume an entire unit of $B$, of which agent 1 needs at least a $1/2 - O(\delta)$ share if we wish to respect individual rationality.
However, these constraints on the prices imply excess demand for $A$ as agent 2 will consume more shares of this good than agent 1: the two agents have the same endowment and hence equal income in the absence of surplus, but whereas agent 1 will spend her income on $A$ and $B$, agent 2 will spend her income on $A$ and $C$, a strictly cheaper good than $B$.

\begin{table}
    \centering
    \subfloat[Utilities]{\begin{tabular}{c|ccc} 
        & $A$ & $B$ & $C$ \\
        \hline
        1 & 2 & 1 & 0 \\
        2 & 2 & 1 & $1 + \delta$ \\
        3 & 0 & 1 & 0
    \end{tabular}}
    \quad\quad\quad\quad
    \subfloat[Endowments]{\begin{tabular}{c|ccc} 
        & $A$ & $B$ & $C$ \\
        \hline
        1 & $1/2$ & $1/2$ & 0 \\
        2 & $1/2$ & $1/2$ & 0 \\
        3 & 0 & 0 & 1 
    \end{tabular}}
    \quad\quad\quad\quad
    \subfloat[Allocation]{\begin{tabular}{c|ccc} 
        & $A$ & $B$ & $C$ \\
        \hline
        1 & $1/2$ & $1/2$ & 0 \\
        2 & $1/2$ & 0 & $1/2$ \\
        3 & 0 & $1/2$ & $1/2$ 
    \end{tabular}}
    \caption{Example with an infinitely more valuable good}
    \label{table:our-example}
\end{table}

The example in Table~\ref{table:our-example} highlights the second obstacle: the presence of goods that are ``infinitely more valuable.''
In the example, $A$ is infinitely more valuable than $B$ and $C$: the unique quasi-equilibrium consists of the allocation in Table~\ref{table:our-example} and the prices $(1,0,0)$.
This suggests that if we want to price the trade of $B$ and $C$ between agents 2 and 3, then we must do so in an infinitely less valuable currency than the currency in which $A$ is priced.

Consider the following pricing scheme. 
\[
    p_A = \begin{pmatrix}
        1 \\
        0
    \end{pmatrix}, 
    p_B = \begin{pmatrix}
        0 \\
        1
    \end{pmatrix}, 
    p_C = \begin{pmatrix}
        0 \\
        0
    \end{pmatrix}
\]
Let $P = (p_A, p_B, p_C)$ denote the 2-by-3 matrix whose columns are given by the vectors above.
Let $p^{(1)}$ denote the first row of $P$ and $p^{(2)}$ the second.
At these ``prices,'' the worth of each agent's endowment is as follows.
\[
    P \omega_1 = P \omega_2 = \begin{pmatrix}
        1/2 \\
        1/2
    \end{pmatrix}, 
    P \omega_3 = \begin{pmatrix}
        0 \\
        0
    \end{pmatrix}
\]
Now, consider the following dividends
\[
    \alpha_1 = \alpha_2 = \begin{pmatrix}
        0 \\
        0
    \end{pmatrix}, 
    \alpha_3 = \begin{pmatrix}
        0 \\
        1/2
    \end{pmatrix}
\]
and suppose we allowed agent $i$ to purchase any lottery $x$ in the closure\footnote{We take the closure because an agent could always exchange $\varepsilon$ units of the first currency for infinite units of the second.} of $\{x \in \Delta^3_- : P x \leq_{\mathrm{lex}} P \omega_i + \alpha_i\}$ where $\leq_{\mathrm{lex}}$ denotes the lexicographic ordering on $\RR^3$.
Since agents 1 and 2 possess positive wealth in the first currency, their budget sets are as follows.
\[
    C^\alpha_1(P) = C^\alpha_2(P) = \{x \in \Delta^3_- : p^{(1)} \cdot x \leq p^{(1)} \cdot \omega_1 + \alpha^{(1)}_1\} = \{x \in \Delta^3_- : x_A \leq 1/2\}
\]
Here, we have used the fact that agents 1 and 2 have the same endowment and dividend.
Meanwhile, since agent 3 has no income in the first currency, her budget set is the following.
\[
    C^\alpha_3(P) = \{x \in \Delta^3_- : p^{(1)} \cdot x \leq p^{(1)} \cdot \omega_3 + \alpha^{(1)}_3 \wedge p^{(2)} \cdot x \leq p^{(2)} \cdot \omega_3 + \alpha^{(2)}_3\} = \{x \in \Delta^3_- : x_A = 0, x_B \leq 1/2\}
\]
With these budget sets, each agent's favorite lottery is precisely the one she receives under the allocation in Table~\ref{table:our-example}, so this allocation with prices $P$ and dividends $\alpha$ clears the market.
We extend this idea into a general solution concept for the problem of one-sided matching with endowments.
We call this solution concept a lexicographic dividend equilibrium (LDE).

We make two observations regarding the equilibrium for the example in Table~\ref{table:our-example}.
The first observation is that because agents 1 and 2 own positive wealth in the first currency, their budget set does not restrict their consumption of $B$ and $C$ \textit{at all}, despite the fact that $B$ has a positive price in the second currency.
In other words, these two goods appear free to agents 1 and 2, so they face the same optimization problem they did at the unique quasi-equilibrium prices.
The second observation is that surplus exists in the second currency, despite the fact that no agent is satiated.
This is because agent 2 is satiated with respect to the second currency: she wants more of $A$ but cannot use her surplus in the second currency make this purchase as $A$ is infinitely more expensive.
As a consequence, we can redistribute this surplus to agent 3 so that she can trade half of her share of $C$ for the entirety of agent 2's share of $B$.

\begin{definition}[Lexicographic dividend equilibrium]
A $d$-dimensional lexicographic dividend equilibrium (LDE) for the economy $(u, \omega)$ consists of an allocation $x \in \mathcal{M}$, a price matrix $p \in \RR^{d \times n}$, and a dividend matrix $\alpha \in \RR^{d \times n}_+$ such that
\begin{itemize}
    \item The existence of $k \in [d]$ such that $p^{(k)}_j < 0$ implies that $p^{(\ell)}_j > 0$ for some $\ell \in [k-1]$.
        That is, the first non-zero entry in $p_j$ is positive for all items $j \in [m]$.
    \item $\alpha^{(k)}_i = \max\{p^{(k)} \cdot (x_i - \omega_i), 0\}$ for all agents $i \in [n]$ and $k \in [d]$, .
    \item $x_i \succeq_i y_i$ for all agents $i \in [n]$ and $y_i \in C^\alpha_i(p)$.
\end{itemize}
where
\begin{align*}
    C^\alpha_i(p) 
        &\coloneqq \{y_i \in \Delta^n_- : p^{(k)} \cdot y_i \leq p^{(k)} \cdot \omega_i + \alpha^{(k)}_i \,\forall\, \ell \in [k_i]\} \\
    k_i 
        &\coloneqq \min \{k \in [d] : p^{(k)} \cdot \omega_i + \alpha^{(k)}_i > 0\} \cup \{d\}
\end{align*}
Additionally, an LDE $(x, p, \alpha)$ satisfies 
\begin{itemize}
    \item Simple prices if $\lvert \mathrm{supp}(p_j) \rvert \leq 1$ for all items $j$.
    \item Non-negative prices if $p \in \RR^{d \times n}_+$.
    Note that simple prices imply non-negative prices since the first non-zero entry in $p_j$ is positive for all items $j$.
    \item The strong cheapest bundle property if there does not exist an agent $i$, a bundle $y_i \succeq_i x_i$, and $k \in [d]$ such that $p^{(\ell)} \cdot y_i = p^{(\ell)} \cdot x_i$ for all $\ell \in [k-1]$ and $p^{(k)} \cdot y_i < p^{(k)} \cdot x_i$.
    \item The weak cheapest bundle property if there does not exist an agent $i$, a bundle $y_i \succeq_i x_i$, and $k \in [k_i]$ such that $p^{(\ell)} \cdot y_i = p^{(\ell)} \cdot x_i$ for all $\ell \in [k-1]$ and $p^{(k)} \cdot y_i < p^{(k)} \cdot x_i$.
    \item The aggregate cheapest bundle property if there does not exist $\beta \in \RR^n_+$, a profile of bundles $y \in \Delta^{n \times n}_-$, and $k \in [d]$ such that 
    \begin{itemize}
        \item $y_i \succeq_i x_i$ for all agents $i \in \supp{\beta}$
        \item $p^{(\ell)} \cdot \sum_i \beta_i y_i = p^{(\ell)} \cdot \sum_i \beta_i x_i$ for all $\ell \in [k-1]$ and $p^{(k)} \cdot \sum_i \beta_i  y_i < p^{(k)} \cdot \sum_i \beta_i x_i$
        \item $\sum_i \beta_i y_{ij} \leq \sum_i \beta_i x_{ij}$ for all items $j$ such that $p^{(\ell)}_j > 0$ for some $\ell \in [k-1]$.
    \end{itemize}
\end{itemize}
\end{definition}

Several remarks are in order.
First, when $d = 1$, the definition of an LDE coincides with that of a dividend equilibrium.
Second, $\sum_i \alpha^{(k)}_i$ is in fact the unspent income in currency $k$ since $\sum_i x_i = \sum_i \omega_i$.
However, unlike \cite{AumannD1986, McLennan2018}, this unspent income need not flow from satiated agents to unsatiated agents.
There are no satiated agents in the LDE for the example in Table~\ref{table:our-example}, yet there exists a positive surplus in the second currency.
In general, the surplus in currency $k$ need not even flow from agents with positive wealth in higher currencies (who have no use for this currency) to those without (who do have use for this currency). 
See Table~\ref{table:d=3-example} for an example where two agents (agents 2 and 3) with positive endowments in the first currency exchange their surplus in the second and third currencies at the LDE.

\begin{table}
    \centering
    \subfloat[Utilities]{\begin{tabular}{c|cccc} 
        & $A$ & $B$ & $C$ & $D$ \\
        \hline
        1 & 2 & 1 & 0 & 0 \\
        2 & 2 & 1 & $1 + \delta$ & 0 \\
        3 & 2 & $1 + \delta$ & 1 & 0 \\
        4 & 2 & 0 & 1 & $1 + \delta$ \\
        5 & 0 & 0 & 1 & 0 \\
        6 & 0 & 1 & 0 & 0
    \end{tabular}}
    \quad\quad
    \subfloat[Endowments]{\begin{tabular}{c|cccc} 
        & $A$ & $B$ & $C$ & $D$ \\
        \hline
        1 & $1/2$ & $1/2$ & 0 & 0 \\
        2 & $1/2$ & $1/2$ & 0 & 0 \\
        3 & $1/2$ & 0 & $1/2$ & 0 \\
        4 & $1/2$ & 0 & $1/2$ & 0 \\
        5 & 0 & 0 & 0 & 1 \\
        6 & 0 & 0 & 0 & 1
    \end{tabular}}
    \quad\quad
    \subfloat[Allocation]{\begin{tabular}{c|cccc} 
        & $A$ & $B$ & $C$ & $D$ \\
        \hline
        1 & $1/2$ & $1/2$ & 0 & 0 \\
        2 & $1/2$ & 0 & $1/2$ & 0 \\
        3 & $1/2$ & $1/2$ & 0 & 0 \\
        4 & $1/2$ & 0 & 0 & $1/2$ \\
        5 & 0 & 0 & $1/2$ & $1/2$ \\
        6 & 0 & 0 & 0 & 1
    \end{tabular}}
    \quad\quad
    \subfloat[Prices]{\begin{tabular}{c|cccc} 
        $k$ & $A$ & $B$ & $C$ & $D$ \\
        \hline
        1 & 1 & 0 & 0 & 0 \\
        2 & 0 & 1 & 0 & 0 \\
        3 & 0 & 0 & 1 & 0
    \end{tabular}}
    \quad\quad\quad\quad
    \subfloat[Dividends]{\begin{tabular}{c|cccccc} 
        $k$ & 1 & 2 & 3 & 4 & 5 & 6 \\
        \hline
        1 & 0 & 0       & 0 & 0 & 0     & 0 \\
        2 & 0 & 0       & $1/2$     & 0 & 0     & 0 \\
        3 & 0 & $1/2$   & 0     & 0 & $1/2$ & 0
    \end{tabular}}
    \caption{Example with $d = 3$}
    \label{table:d=3-example}
\end{table}

It is well-known that the first welfare theorem need not hold, that is, the allocation may be inefficient, without some sort of cheapest bundle property \cite{Eisenberg1961, HyllandZ1979, Mas-Colell1992, BudishCKM2013, Le2017, EcheniqueMZ2021Constrained, EcheniqueMZ23}.
We give three cheapest bundles properties.
The strong cheapest bundle property requires that a weakly preferred bundle that costs lexicographically less does not exist and implies the other two properties.
The weak cheapest bundle property only requires that there does not exist a weakly preferred bundle that costs strictly less in the most valuable currency that an agent owns a positive quantity of.
The existence of a weakly preferred bundle that costs the same in this currency but strictly less in a lower currency violates the strong cheapest bundle property but not its weak counterpart.
Meanwhile, the aggregate cheapest bundle property requires that no group of agents can redistribute their collective allocation and spend lexicographically less without making some member worse off, even in larger economies.
The weak and aggregate cheapest bundle properties are disjoint notions.
On its own, the strong cheapest bundle property suffices for the first welfare theorem to hold.
In fact, it suffices for an LDE to reside within the rejective core.
We will see that the weak and aggregate versions, together with the simple prices property, serve the same function.

Our first main result is that an LDE satisfying the strong cheapest bundle property always exists for any one-sided matching market with endowments.

\begin{restatable}{theorem}{LDEexistence}\label{theorem:LDE-existence}
For all economies $(u, \omega)$, an LDE satisfying the strong cheapest bundle property exists.
\end{restatable}

We defer the proof to the appendix but offer some intuition using the example in Table~\ref{table:our-example}.
Recall that the price and dividend matrices of the LDE are given by
\[
    p_A = \begin{pmatrix}
        1 \\
        0
    \end{pmatrix}, 
    p_B = \begin{pmatrix}
        0 \\
        1
    \end{pmatrix}, 
    p_C = \begin{pmatrix}
        0 \\
        0
    \end{pmatrix},
    \alpha_1 = \alpha_2 = \begin{pmatrix}
        0 \\
        0
    \end{pmatrix}, 
    \alpha_3 = \begin{pmatrix}
        0 \\
        1/2
    \end{pmatrix}.
\]
The main idea of the proof is to examine a sequence of economies $(u, \omega^\varepsilon)_\varepsilon$ that approach the economy $(u, \omega)$.
In the economy $(u, \omega^\varepsilon)$, we have redistributed an $\varepsilon$ fraction of the goods so that a 1-dimensional LDE exists.
In the proof, we redistribute an $\varepsilon / n$ fraction of each good, but to keep the analysis here simple, we will instead exchange an $\varepsilon$ share of $A$ from each of agents 1 and 2 with a $2 \varepsilon$ share of $C$ from agent 3 as described in Table~\ref{table:eps-economy}.
Note that
\[
    \omega^\varepsilon = (1 - 2\varepsilon) \cdot \omega + 2\varepsilon \cdot \begin{pmatrix}
        0 & 1/2 & 1/2 \\
        0 & 1/2 & 1/2 \\
        1 & 0 & 0
    \end{pmatrix}
\]
For all $\varepsilon > 0$, the allocation $x^\varepsilon$ in Table~\ref{table:eps-economy} and the prices $p^\varepsilon = (1, 4\varepsilon, 0)$ constitute a competitive equilibrium.\footnote{This is another sense in which $A$ is infinitely more valuable than $B$ and $C$: simply owning any positive share of this good enables the existence of a dividend equilibrium. In this sense, $A$ acts as a medium of exchange.}
Note that $x^\varepsilon$ converges to the LDE allocation $x$ given in Table~\ref{table:our-example}.

\begin{table}
    \centering
    \subfloat[Utilities]{\begin{tabular}{c|ccc} 
        & $A$ & $B$ & $C$ \\
        \hline
        1 & 2 & 1 & 0 \\
        2 & 2 & 1 & $1 + \delta$ \\
        3 & 0 & 1 & 0
    \end{tabular}}
    \quad\quad\quad\quad
    \subfloat[Endowments]{\begin{tabular}{c|ccc} 
        & $A$ & $B$ & $C$ \\
        \hline
        1 & $1/2 - \varepsilon$ & $1/2$ & $\varepsilon$ \\
        2 & $1/2 - \varepsilon$ & $1/2$ & $\varepsilon$ \\
        3 & $2\varepsilon$ & 0 & $1-2\varepsilon$ 
    \end{tabular}}
    \quad\quad\quad\quad
    \subfloat[Allocation]{\begin{tabular}{c|ccc} 
        & $A$ & $B$ & $C$ \\
        \hline
        1 & $1/2 - \varepsilon$ & $1/2$ & $\varepsilon$ \\
        2 & $1/2 + \varepsilon$ & 0 & $1/2 - \varepsilon$ \\
        3 & 0 & $1/2$ & $1/2$ 
    \end{tabular}}
    \caption{The prices $p^\varepsilon = (1, 4\varepsilon, 0)$ clear the economy $(u, \omega^\varepsilon)$}
    \label{table:eps-economy}
\end{table}

To recover the price and dividend matrices of the LDE, we examine the order of the entries in $p^\varepsilon$ and elect a representative commodity from each order.
In this case, we see that $p^\varepsilon_A$ is $\Omega(1)$ while $p^\varepsilon_B$ and $p^\varepsilon_C$ are $O(\varepsilon)$.
Our proof would elect $A$ and $B$ as the representatives of their respective orders.
The first row of the price matrix $p^{(1)}$ is then defined as
\[
    p^{(1)} = \lim_{\varepsilon \to 0} p^\varepsilon / p^\varepsilon_A = (1, 0, 0)
\]
while the second row is defined as
\[
    p^{(2)} = \lim_{\varepsilon \to 0} (p^\varepsilon - p^{(1)}) / p^\varepsilon_B = (0, 1, 0)
\]
In other words, we essentially decompose $p^\varepsilon$ into its principal components.
Note that this process recovers the price matrix of the LDE.
We recover the dividend matrix as follows.
For all but the last currency, the entries of the $k$-th row of the dividend matrix are simply taken to be 
\[
    \alpha^{(k)}_i = \max\{p^{(k)} \cdot (x_i - \omega_i), 0\},
\]
while the last row of the dividend matrix is instead defined as
\[
    \alpha^{(2)} = \lim_{\varepsilon \to 0} \, (\omega^\varepsilon - (1 - 2\varepsilon) \cdot \omega) \, p^\varepsilon / p^\varepsilon_B = (0, 0, 1/2)
\]
That is, the worth of the $\varepsilon$ redistribution with respect to the representative of the last currency determines the surplus in this currency.\footnote{Because in the proof, we redistribute a $\varepsilon / n$ share of each good to every agent, each agent will have a positive dividend in the last currency. This dividend technically does not satisfy the condition of an LDE.
However, some agents have no use for extra income in the last currency, so we can lower their dividends so that the required condition is satisfied.}
We then use the fact that $(x^\varepsilon, p^\varepsilon)$ constitutes a 1-dimensional LDE in the economy $(u, \omega^\varepsilon)$ to argue that $(x, p, \alpha)$ constitutes an LDE in the economy $(u, \omega)$.
Some additional care is required to guarantee the strong cheapest bundle property.

\section{Core Convergence}

A compelling argument for the notion of competitive equilibrium as a solution concept is that competitive allocations are the only viable outcomes in a very strong sense: they are the only outcomes that a coalition cannot weakly block as the economy grows.
\cite{DebreuS1963} formalize this idea using replica economies.
We demonstrate a similar result for the notion of lexicographic dividend equilibrium in one-sided matching markets with endowments.
More specifically, we show that LDEs are the only outcomes that a coalition cannot reject in a large market.

\begin{definition}[Replica economy]
The $N$-fold replica economy of $(u, \omega)$ is defined as the economy $(u^{(m)}, \omega^{(m)})_{m \in [N]}$, that is, the economy with $N$ replicas of each agent in $(u, \omega)$.
For notational simplicity, we simply denote the $N$-fold replica economy as $(u, \omega)_N$.
We define $LDE_+(u, \omega)$ to be the set of LDEs in $(u, \omega)$ satisfying the strong cheapest bundle property and $RC(u, \omega, N)$ to be the set of rejective core stable allocations in $(u, \omega)_N$ that allocate the same bundle to each replica of each agent.
\end{definition}

Our definition for the set of core allocations in a replica economy slightly departs from the conventions established by \cite{DebreuS1963}.
Their assumptions of insatiability and strict convexity imply that elements of the weak core always allocate each replica of an agent the same bundle.
Because we have neither of these assumptions, it is possible for an agent's replicas to receive distinct bundles, in which case the elements of $RC(u, \omega, N)$ would reside in a higher dimension than the elements of $LDE_+(u, \omega)$, and the two sets would be incomparable.
Our second main result is that the set of LDEs satisfying the strong cheapest bundle property coincides with the set of rejective core stable allocations in large economies.

\begin{theorem}\label{theorem:core-convergence}
$LDE_+(u, \omega) = \bigcap_{N=1}^\infty RC(u, \omega, N)$ for all economies $(u, \omega)$.
\end{theorem}

\subsection{The rejective core stability of LDEs in large economies}

The main result of this section is that the set of rejective core stable allocations in large economies contains the set of LDEs satisfying the strong cheapest bundle properties.

\begin{restatable}{theorem}{rejectiveCoreContainsLDEs}\label{theorem:rejective-core-contains-LDEs}
$LDE_+(u, \omega) \subseteq \bigcap_{N=1}^\infty RC(u, \omega, N)$ for all economies $(u, \omega)$.
\end{restatable}

We defer a formal proof to Section~\ref{section:proof-rejective-core-contains-LDEs}, but the idea is quite simple.
Let $(x, p, \alpha) \in LDE_+(u, \omega)$, and consider any coalition $C = (C_1, C_2)$ that rejects this allocation because its members collectively prefer the allocation $y$, which they can reach by redistributing the collective endowment of $C_1$ and the collective allocation of $C_2$.
Since $y$ is a redistribution of $\omega_{C_1}$ and $x_{C_2}$, we have that 
\[
    \sum_{i \in C} p y_i \leq_{\mathrm{lex}} \sum_{i \in C_1} p \omega_i + \sum_{i \in C_2} p x_i.
\]
However, the strong cheapest bundle property, together with the fact that at least one member of the coalition is strictly happier, implies the opposite inequality.
Some care is needed to formalize this intuition as LDE prices may be negative in some currencies (but never in the first non-zero coordinate) and the definition of a rejecting coalition allows it to freely dispose of the goods it possesses, either through the initial endowments or the allocations of its members.\footnote{At this point, the reader may have concerns about certain economies such as those in which a single agent is endowed with all the goods in the market. Our results imply that the rejective core is non-empty for this economy, but it appears as though this agent forms a rejecting coalition with any other agent if we allow free disposal. However, note that the definition of a rejecting coalition only permits the inclusion of this agent if she brings her allocation, as she can only bring her endowment if she can be made strictly happier, which is not possible.}

\subsection{The rejective core converges to the set of LDEs}

In this section, we prove the reverse containment: the set of LDEs satisfying the strong cheapest bundle property contains the rejective core regardless of the size of the economy.

\begin{theorem}
$LDE_+(u, \omega) \supseteq \bigcap_{N=1}^\infty RC(u, \omega, N)$ for all economies $(u, \omega)$.
\end{theorem}

The proof is broken down into two overarching steps.
The first step shows that the set of LDEs with simple prices satisfying the weak and aggregate cheapest bundle properties contains the rejective core (Theorem~\ref{theorem:weak-LDE-contains-rejective-core}).
We use $LDE_-(u, \omega)$ to denote this set of LDEs and in the second step, prove that it resides in $LDE_+(u, \omega)$ (Theorem~\ref{theorem:from-weak-aggregate-to-strong}).

\begin{restatable}{theorem}{weakLDEcontainsRC}\label{theorem:weak-LDE-contains-rejective-core}
$LDE_-(u, \omega) \supseteq \bigcap_{N=1}^\infty RC(u, \omega, N)$ for all economies $(u, \omega)$.
\end{restatable}

We defer the proof to the appendix and opt for an intuitive argument here.
At a high level, the proof recursively finds a quasi-equilibrium among the goods that have yet to be priced and continues to do so until either every good has received a positive price in some currency or the agents with no income in any currency receive their favorite free bundles.

Let $x \in \bigcap_{N=1}^\infty RC(u, \omega, N)$.
We outline the second step of the recursion because the first step can be seen as a special case (and uses more standard ideas).
Moreover, doing so allows us to convey the main ideas without introducing the cumbersome notation needed to describe a general recursive step.
With that aside, suppose we have found $p^{(1)} \in \RR^n_+$ such that $(x, p^{(1)})$ constitutes an expenditure-minimizing quasi-equilibrium, and let us turn to the second recursive step.

As in \cite{DebreuS1963}, we seek a hyperplane, which we will ultimately take to be $p^{(2)}$, that separates the set of preferred trades from 0 simultaneously for each agent. 
However, if we naively define these sets (and for that matter, the set we wish to separate them from), then it is possible that the only separating hyperplane is $p^{(1)}$, in which case no progress is made.
Thus, we need to carefully define the sets of interest so that 
\begin{enumerate}
    \item at least one good with zero price in $p^{(1)}$ receives a positive price in $p^{(2)}$
    \item each agent with no income with respect to $p^{(1)}$ but positive income with respect to $p^{(2)}$ receives her favorite affordable bundle that minimizes her expenditure in the second currency
    \item there does not exist a ``coalition'' consisting of a $\beta_i$ share of each agent $i$ with positive income in the first currency that can redistribute its collective allocation so that each agent in the support of $\beta$ is weakly better off, but the coalition spends strictly less with respect to $p^{(2)}$
\end{enumerate}
Note that we require the third condition for the aggregate cheapest bundle property.

Rather than separate the set of preferred trades (whose definition we defer for now) from 0, we separate from the following set:
\[
    V \coloneqq \{z \in \RR^n : z_{j} < 0 \,\forall\, j \in S\}
\]
where $S$ denotes the goods with zero price in $p^{(1)}$.
We separate from $V$ for two reasons.
The first is that our definition for the set of preferred trades will in fact contain 0 as an element.
Without disjoint-ness, we cannot apply the separating hyperplane theorem.
The second reason is that we need the hyperplane to assign some item in $S$ a positive price.
The most straightforward way to ensure that this happens is to make $V$ unconstrained in the coordinates corresponding to the goods in $S$.
This way, the separating hyperplane cannot price the goods in $S$ in the second currency.

Now, let $T$ denote the agents with positive income in the first currency.
For the agents not in $T$, we define
\begin{align*}
    P_i 
        &\coloneqq \{z_i \in \RR^n : z_i + \omega_i \succ_i x_i \wedge z_{ij} = 0 \,\forall\, j \not\in S\} \\
    Q_i
        &\coloneqq \{z_i \in \RR^n : z_i + x_i \succeq_i x_i \wedge z_{ij} = 0 \,\forall\, j \not\in S\}
\end{align*}
With the exception of the requirement that $z_{ij} = 0$ for all $j \not\in S$, the definition of these sets is rather standard.
Note that there is no need to consider trades that involve the goods not in $S$ since these goods are priced in the first currency and hence inaccessible to the agents not in $T$ anyway.
In particular, if $p^{(2)}$ separates $P_i$ from $V$, then we can show that agent $i \not\in T$ receives her favorite affordable bundle if given a dividend of $\alpha^{(2)}_i \coloneqq \max\{p^{(2)} \cdot (x_i - \omega_i), 0\}$.
Moreover, if $p^{(2)}$ separates $Q_i$ from $V$ as well, then we can show that this bundle also minimizes her expenditure in the second currency.
The extra degree of freedom from restricting the set of preferred trades to those that do not trade any goods not in $S$ is one reason why a separating hyperplane aside from $p^{(1)}$ might exist.

More care is required to properly define the set of preferred trades for the agents in $T$ if we wish to obtain the third condition listed while giving the hyperplane as much wiggle room as possible.
First, for each agent in $T$, let
\[
    Q_i \coloneqq \{z_i \in \RR^n : z_i + x_i \succeq_i x_i\}.
\]
The ``correct'' set of preferred trades to consider for the agents in $T$ is then 
\[
    Q_T \coloneqq \{z \in \RR^n : z_{j} \leq 0 \,\forall\, j \not\in S\} \cap \mathrm{conv} \cup_{i \in T} Q_i
\]
where $\mathrm{conv}$ denotes the convex hull.
Note that an element of $Q_T$ can be expressed as $\sum_{i \in T} \beta_i (y_i - x_i)$ where $\beta \in \Delta^T$ and $(y_i - x_i) \in Q_i$ (so $y_i \succeq_i x_i$) for all $i \in T$. 
Moreover, because we intersect with the set of $z$ such that $z_j \leq 0$ for all $j \not\in S$, we have that
\[
    \sum_{i \in T} \beta_i (y_{ij} - x_{ij}) \leq 0 \,\forall\, j \not\in S
\]
so $y$ is a redistribution (at least with respect to the goods not in $S$) of $x$ among a ``coalition'' that consists of a $\beta_i$ share of each agent $i \in T$.
It is no coincidence that the elements of $Q_T$ have this interpretation: $Q_T$ is purposefully defined to recover the aggregate cheapest bundle property.
If $p^{(2)}$ separates $Q_T$ from $V$, then it must be that $y$ collectively costs at least as much as $x$ in the second currency for this coalition.

Note that had we defined $Q_i$ for $i \in T$ analogously to $Q_i$ for $i \not\in T$ and separated this set from $V$, then we would only be able to say that the agents with positive wealth in the first currency are minimizing their expenditure in the second currency \textit{fixing their consumption of the goods not in $S$}.
On the other hand, had we worked with $Q_i$ as defined for $i \in T$ directly, then there may not be enough wiggle room for a separating hyperplane other than $p^{(1)}$ to exist. 

The set of preferred trades we wish to separate from $V$ is then
\[
    P \coloneqq \mathrm{conv} \, (\cup_{i \not\in T} (P_i \cup Q_i) \cup Q_T).
\]
It remains to show that a separating hyperplane exists.
As $P$ and $V$ are non-empty and convex, it suffices to show that they are disjoint.
Toward this end, we appeal to the fact that $x$ resides within the rejective core of every replica economy: if the two sets are not disjoint, then from a point in the intersection, we can recover a rejecting coalition in a sufficiently large market.
This part of the argument is standard, but we highlight one subtlety.
The definition of rejecting coalition requires that the allocation $y$ collectively desired by the coalition is a redistribution of some members' endowments and other members' allocations.
However, $V$ is unconstrained in the coordinates corresponding to the goods not in $S$.
The fact that these goods are not over-demanded instead comes from the definition of $P$: $Q_T$ ensures that the agents in $T$ demand only as much as the coalition owns, while $P_i$ and $Q_i$ for $i \not\in T$ ensure that these agents continue to consume only goods in $S$.

If there remain goods yet to be priced and agents with no income in either the first or second currencies that do not receive their favorite free bundles, then the recursion continues.
In the proof of Theorem~\ref{theorem:weak-LDE-contains-rejective-core}, we define a formal process that recovers an LDE with simple prices satisfying the weak and aggregate cheapest bundle properties and prove its correctness and termination by induction.
Each recursive step and its analysis resemble the second recursive step and what we discussed above.

\begin{restatable}{theorem}{fromWeakToStrong}\label{theorem:from-weak-aggregate-to-strong}
$LDE_+(u, \omega) \supseteq LDE_-(u, \omega)$.
That is, the set of LDEs satisfying the strongest bundle property contains the set of LDEs with simple prices that satisfy the weak and aggregate cheapest bundle properties.
\end{restatable}

The reader can find the complete proof in Section~\ref{section:proof-from-weak-to-strong}.
The main idea is to modify the price vector in each currency so that the ultimate price matrix satisfies the strong cheapest bundle property. 
Let $(x, p, \alpha) \in LDE_-(u, \omega)$.
Since $p$ is simple, each good is priced (positively) in at most one currency.
Let $S_k$ denote the set of items that are priced in currencies $k$ through $d$ (so any good not in $S_k$ has a positive price in a more valuable currency), and let $T_k$ denote the set of agents with positive income in a currency strictly more valuable than currency $k$.
Note that modifying the price in currency $k$ of a good not in $S_k$ does not affect the demand of any agent: for the agents in $T_k$, only the prices in more valuable currencies are relevant, and agents not in $T_k$ cannot afford these goods anyway since they do not possess positive wealth in currencies more valuable than $k$.
Thus, we can take advantage of this degree of freedom to modify these prices to ensure that the agents in $T_k$ do not possess bundles weakly preferred to their allocations that cost strictly less in currency $k$.
Note that this suffices for the strong cheapest bundle property.

We show that such a modification exists by setting up a system of linear constraints and proving that it has a feasible solution.
In particular, we argue that infeasibility implies a violation to the aggregate cheapest bundle property of $(x, p, \alpha)$.
We then use the feasible solution to construct new price and dividend matrices, $q$  and $\gamma$, respectively.
Since we only modify the prices inconsequential to each agent's demand, we get that $(x, q, \gamma)$ constitutes an LDE satisfying the strong cheapest bundle property, that is, $(x, q, \gamma) \in LDE_+(u, \omega)$.

\section{Discussion}

In this paper, we propose the solution concept of a lexicographic dividend equilibrium for the problem of one-sided matching with endowments.
We show that an LDE always exists and that the set of LDEs coincides with the rejective core as the economies grows, which suggests that equilibria of this kind are the only viable market outcomes in some sense.
Our result completely circumvents the non-existence of a competitive (and dividend) equilibrium in economies with satiation but does so in a way that does not sacrifice desirable properties, such as individual rationality and weak core stability.

We leave the reader with two open questions.
The first question is whether two currencies always suffice to clear the market.
While the example given in Table~\ref{table:d=3-example} involves three currencies, it is not difficult to see that there exists an LDE that involves only two.
Table~\ref{table:d=3-example-but-actually-d=2} constitutes such an LDE and also serves as an example that prices more than one item in the first currency.
Intuitively, two currencies should suffice: the agents who possess positive wealth in the first currency consume their favorite bundles and discard the rest, inducing a surplus in the second currency that can be redistributed so that the remaining agents all have positive wealth in this currency.
Unfortunately, our proof approach does not suffice to validate this conjecture because it fails to distinguish between the LDEs given in Table~\ref{table:d=3-example} and Table~\ref{table:d=3-example-but-actually-d=2}.
Note that by our limiting argument in Theorem~\ref{theorem:LDE-existence}, it would suffice to prove that in an economy that redistributes an $\varepsilon$ share of each good, there exists a competitive (or dividend) equilibrium whose prices are either $\Theta(1)$ or $O(\varepsilon)$.

\begin{table}
    \centering
    \subfloat[Allocation]{\begin{tabular}{c|cccc} 
        & $A$ & $B$ & $C$ & $D$ \\
        \hline
        1 & $1/2$ & $1/2$ & 0 & 0 \\
        2 & $\frac{1+\delta}{2}$ & 0 & $\frac{1-\delta}{2}$ & 0 \\
        3 & $\frac{1-\delta}{2}$ & $1/2$ & $\delta/2$ & 0 \\
        4 & $1/2$ & 0 & 0 & $1/2$ \\
        5 & 0 & 0 & $1/2$ & $1/2$ \\
        6 & 0 & 0 & 0 & 1
    \end{tabular}}
    \quad\quad
    \subfloat[Prices]{\begin{tabular}{c|cccc} 
        $k$ & $A$ & $B$ & $C$ & $D$ \\
        \hline
        1 & 1 & $\delta$ & 0 & 0 \\
        2 & 0 & 0 & 1 & 0 
    \end{tabular}}
    \quad\quad
    \subfloat[Dividends]{\begin{tabular}{c|cccccc} 
        $k$ & 1 & 2 & 3 & 4 & 5 & 6 \\
        \hline
        1 & 0 & 0   & 0     & 0 & 0         & 0 \\
        2 & 0 & $\frac{1-\delta}{2}$   & 0     & 0 & $1/2$     & 0 
    \end{tabular}}
    \caption{Example with $d = 2$}
    \label{table:d=3-example-but-actually-d=2}
\end{table}

The second open question regards the extent to which our results extend beyond the setting of one-sided matching with endowments.
\cite{AumannD1986, McLennan2018} study more general settings and give sufficient conditions for the existence of a 1-dimensional LDE in an economy that has redistributed an $\varepsilon$ share of each good, so one could presumably apply our limiting argument to obtain an LDE for settings where their results hold.
However, our argument for the strong cheapest bundle property requires linearity of preferences.
In particular, we appeal to linearity to argue that if $x \sim_i x'$, then $x^\varepsilon + (x' - x) / 2 \succeq_i x^\varepsilon$ for all $x^\varepsilon$ in a sufficiently small neighborhood of $x$.
That is, exchanging part of the bundle with another part that an agent is indifferent between should not make her strictly worse off.
Note that convexity of preferences is in some sense almost enough to guarantee this condition as $x^\varepsilon$ and $x$ are arbitrarily close.

Moreover, our proof that the limit of the rejective core resides within the set of LDEs with simple prices that satisfy the weak and aggregate (Theorem~\ref{theorem:weak-LDE-contains-rejective-core}) only requires continuity and convexity of preferences.
However, our argument that the latter set resides within the set of LDEs satisfying the strong cheapest bundle property (Theorem~\ref{theorem:from-weak-aggregate-to-strong}) requires that the preferences be linear and that the consumption set be a polyhedron.
Otherwise, the system of linear constraints that describe the set of feasible price modifications may require an infinite number of constraints to describe, in which case strong duality may fail to hold.

\section{Proof of Theorem~\ref{theorem:rejective-core-contains-LDEs}}\label{section:proof-rejective-core-contains-LDEs}

For the proofs in this section, it will be helpful to refer to the coordinate $k^j$ of the first non-zero entry in $p_j$ (if it exists).
Recall from the definition of LDE that this entry must be positive.
\[
    k^j \coloneqq \min \{k \in [d] : p^{(k)}_j \not= 0\} \cup \{d\}
\]

\begin{lemma}\label{lemma:no-ownership-of-higher-tiered-items-unequal-surplus}
Let $(x, p, \alpha) \in LDE_+(u, \omega)$.
$x_{ij} = \omega_{ij} = 0$ for all agents $i$ and items $j$ such that $k^j < k_i$. 
\end{lemma}

\begin{proof}
Let $\ell \coloneqq \min \{k^j: x_{ij} > 0 \vee \omega_{ij} > 0\}$, and suppose by way of contradiction that $\ell < k_i$.
Note that $p^{(\ell)}_j > 0$ for all items $j$ such that $k^j = \ell$ (since in particular, $\ell \in [d - 1]$), so we know that either $x_{ij} > 0$ or $\omega_{ij} > 0$ for some item $j$ with $p^{(\ell)}_j > 0$.
Moreover, since the first non-zero entry in $p_j$ is positive for all items $j$, the minimality of $\ell$ implies that $x_{ij} = \omega_{ij} = 0$ for all items $j$ such that $p^{(\ell)}_j < 0$.
It follows that either $p^{(\ell)} \cdot \omega_i + \alpha^{(\ell)}_i \geq p^{(\ell)} \cdot x_i > 0$ or $p^{(\ell)} \cdot \omega_i > 0$, but either inequality would contradict the minimality of $k_i$.
\end{proof}

\rejectiveCoreContainsLDEs*

\begin{proof}
Let $(x, p, \alpha) \in LDE_+(u, \omega)$, and suppose by way of contradiction that $x$ does not reside within the intersection.
Then, for some replica economy $(u, \omega)_N$, there exist coalitions $C$ and $C'$ and an alternative allocation $y$ such that $\sum_{i \in C \cup C'} y_i \leq \sum_{i \in C} \omega_i + \sum_{i \in C'} x_i$ yet 
\begin{itemize}
    \item $y_i \succ_i x_i$ for all $i \in C$.
    \item $y_i \succeq_i x_i$ for all $i \in C'$.
    \item There exists $i \in C \cup C'$ such that $y_i \succ_i x_i$.
\end{itemize}
Since agent $i$ weakly prefers $x_i$ over her other affordable bundles, it must be that $y_i \not\in C^\alpha_i(p)$ for all $i \in C \cup C'$ such that $y_i \succ_i x_i$.
That is, for all such $i$, there exists $k \in [k_i]$ such that 
\[
    p^{(k)} \cdot y_i > p^{(k)} \cdot \omega_i + \alpha^{(k)}_i \geq p^{(k)} \cdot x_i.
\]
Let 
\[
    \ell \coloneqq \min \{k \in [d] : \exists\,i \in C \cup C' : p^{(k)} \cdot y_i > p^{(k)} \cdot x_i\},
\]
so $p^{(k)} \cdot y_i \leq p^{(k)} \cdot x_i$ for all $i \in C \cup C'$ and $k \in [\ell-1]$.
Note that as a result, 
\begin{equation}\label{equation:strong-cheapest-bundle-implication-LDE-in-intersection}
    p^{(k)} \cdot (y_i - x_i) \geq 0 \,\forall\, k \in [\ell]
\end{equation}
by the strong cheapest bundle property with equality for all $k \in [\ell - 1]$.
It follows that
\begin{equation}\label{equation:coalition-spending-in-higher-tier-equal}
    \sum_{i \in C \cup C'} p^{(k)} \cdot y_i 
        \geq \sum_{i \in C \cup C'} p^{(k)} \cdot x_i \,\forall\, k \in [\ell] 
\end{equation}
with equality for all $k \in [\ell-1]$ and strict equality for $\ell$ (by definition of $\ell$).
% We claim that the inequality is in fact strict for $\ell$.
% Otherwise, there exists $i \in C \cup C'$ such that $p^{(\ell)} \cdot (y_i - x_i) < 0$ since there exists an agent for which the opposite strict inequality holds by the definition of $\ell$.
% This contradicts the strong cheapest bundle property.
% To see why, note that for all $i \in C$ with $k_i \in [\ell]$, the minimality of $\ell$ implies that 
% \begin{equation*}
%     p^{(\ell)} \cdot (y_i - \omega_i) > \alpha^{(\ell)}_i \geq 0.
% \end{equation*}
% Meanwhile, Lemma~\ref{lemma:no-ownership-of-higher-tiered-items-unequal-surplus} and Equation~\ref{equation:strong-cheapest-bundle-implication-LDE-in-intersection} imply that 
% \begin{equation*}
%     p^{(\ell)} \cdot (y_i - \omega_i) = p^{(\ell)} \cdot (y_i - x_i) \geq 0.
% \end{equation*}
% for all $i \in C$ with $k_i \not\in [\ell]$.

Now, let $S$ denote the set of items $j$ with $k^j \in [\ell - 1]$.
That is, each item $j \in S$ has its first non-zero entry prior to the $\ell$-th coordinate.
Note that by Lemma~\ref{lemma:no-ownership-of-higher-tiered-items-unequal-surplus}, 
\begin{equation}\label{equation:lower-tier-consumer-consume-no-high-tier-goods}
    \omega_{ij} = 0 \,\forall\, i \in C, j \in S    
\end{equation}
since $y_i \not\in C^\alpha_i(p)$ for all such $i$, so $\ell \leq \min \{k_i : i \in C\}$.
We argue that 
\begin{equation}\label{equation:high-tier-consumers-consume-all-high-tier-items}
    \sum_{i \in C'} y_{ij} = \sum_{i \in C'} x_{ij} \,\forall\, j \in S
\end{equation}
To see why, first note that for all such $j$,
\[
    \sum_{i \in C'} y_{ij} \leq \sum_{i \in C \cup C'} y_{ij} \leq \sum_{i \in C} \omega_{ij} + \sum_{i \in C'} x_{ij} = \sum_{i \in C'} x_{ij}
\]
The second inequality follows from fact that $y$ is a redistribution, and the equality follows from Equation~\ref{equation:lower-tier-consumer-consume-no-high-tier-goods}.
Suppose by way of contradiction that $\sum_{i \in C'} y_{ij} < \sum_{i \in C'} x_{ij}$ for some $j \in S$.
Let $S'$ denote the set of items in $S$ for which this inequality holds, and let
\[
    h \coloneqq \min \{k^j : j \in S'\}
\]
Note that minimality of $h$ implies that $\sum_{i \in C'} y_{ij} = \sum_{i \in C'} x_{ij}$ for all $j \in S'$ with $k^j \in [h-1]$.
Moreover, the definition of $S'$ implies that  $\sum_{i \in C'} y_{ij} = \sum_{i \in C'} x_{ij}$ for all $j \in S \setminus S'$.
It follows that
\begin{align*}
    \sum_{i \in C'} p^{(h)} \cdot (y_i - x_i) 
        &= \sum_{j \in S} \sum_{i \in C'} p^{(h)}_j \cdot (y_{ij} - x_{ij}) \tag{$p^{(h)}_j = 0 \,\forall\, j \not\in S$} \\
        &= \sum_{\substack{j \in S' : \\ k^j \geq h}} \sum_{i \in C'} p^{(h)}_j \cdot (y_{ij} - x_{ij}) \\
        &= \sum_{\substack{j \in S' : \\ k^j = h}} \sum_{i \in C'} p^{(h)}_j \cdot (y_{ij} - x_{ij}) \tag{$p^{(h)}_j = 0 \,\forall\, j: k^j > h$}  \\
        &< 0 \tag{definition of $S'$; $p^{(h)}_j > 0 \,\forall\, j : k^j = h$}
\end{align*}
By Equation~\ref{equation:coalition-spending-in-higher-tier-equal}, this inequality implies that
\[
    \sum_{i \in C} p^{(h)} \cdot y_i > \sum_{i \in C} p^{(h)} \cdot x_i,
\]
so there exists $i \in C$ such that $p^{(h)} \cdot y_i > p^{(h)} \cdot x_i$, which contradicts the minimality of $\ell$ since $h < \ell$.

We now show that 
\begin{equation}\label{equation:coalition-spending-in-current-tier-greater}
    \sum_{i \in C \cup C'} p^{(\ell)} \cdot y_i 
        > \sum_{i \in C} p^{(\ell)} \cdot \omega_i + \sum_{i \in C'} p^{(\ell)} \cdot x_i. 
\end{equation}
Since the inequality in Equation~\ref{equation:coalition-spending-in-higher-tier-equal} is strict for $\ell$, it suffices to show that $p^{(\ell)} \cdot (x_i - \omega_i) \geq 0$ for all $i \in C$.
If not, then $k_i \in [\ell]$.
Note that the minimality of $\ell$, together with the fact that $y_i \succ_i x_i$, implies that $k_i = \ell$.
However, $x_{ij} = 0$ for all $j \in S$ by Lemma~\ref{lemma:no-ownership-of-higher-tiered-items-unequal-surplus} and $y_{ij} = 0$ for all $j \in S$ by Equations~\ref{equation:lower-tier-consumer-consume-no-high-tier-goods} and~\ref{equation:high-tier-consumers-consume-all-high-tier-items} (and the fact that $y$ is a redistribution), so there exists $t \in (0,1)$ such that $ty_i + (1-t)x_i$ is affordable but strictly preferred over $x_i$, contradicting the fact that $x_i$ is agent $i$'s favorite bundle.

To conclude, note that Equations~\ref{equation:lower-tier-consumer-consume-no-high-tier-goods} and~\ref{equation:high-tier-consumers-consume-all-high-tier-items}, coupled with the fact that $y$ is a redistribution, that is, $\sum_{i \in C \cup C'} y_{i} \leq \sum_{i \in C} \omega_{i} + \sum_{i \in C'} x_{i}$, imply that $y_{ij} = 0$ for all $i \in C, j \in S$.
It now readily follows that
\[
     \sum_{i \in C} p^{(\ell)} \cdot (y_i - \omega_i) + \sum_{i \in C'} p^{(\ell)} \cdot (y_i - x_i) 
        = \sum_{j \not\in S} p^{(\ell)}_j \cdot \left(\sum_{i \in C \cup C'} y_{ij} - \sum_{i \in C} \omega_{ij} - \sum_{i \in C'} x_{ij} \right) \\
        \leq 0 
\]
The equality uses Equations~\ref{equation:lower-tier-consumer-consume-no-high-tier-goods} and~\ref{equation:high-tier-consumers-consume-all-high-tier-items}.
The inequality uses the fact that $y$ is a redistribution and the fact that $p^{(\ell)}_j \geq 0$ for all $j \not\in S$.
However, note that this final inequality contradicts Equation~\ref{equation:coalition-spending-in-current-tier-greater}.
\end{proof}

\section{Proof of Theorem~\ref{theorem:from-weak-aggregate-to-strong}}\label{section:proof-from-weak-to-strong}

\fromWeakToStrong*

\begin{proof}
Let $(x, p, \alpha) \in LDE_-(u, \omega)$.
We show that there exists prices $q$ and dividends $\gamma$ such that $(x, q, \gamma) \in LNE_+(u, \omega)$.
For all $k \in [d]$, let
\begin{align*}
    S_k 
        &\coloneqq \{j \in [n] : p^{(\ell)}_j = 0 \,\forall\, \ell \in [k-1]\} \\
    T_k 
        &= \{i \in [n] :  \exists\,\ell \in [k-1] : p^{(\ell)} \cdot \omega_i + \alpha^{(\ell)}_i > 0\}
\end{align*}
We argue that for all $k \in [d]$, there exists $r^{(k)}$ such that
\begin{align*}
    & r^{(k)}_j = 0 \,\forall\, j  \in S_k \\
    & (r^{(k)} + p^{(k)}) \cdot (y - x_i) \geq 0 \,\forall\, i \in T_k, y \in C^\alpha_i(p) : y \succeq_i x_i
\end{align*}
Note that this linear program can be made finite since $\{y \in C^\alpha_i(p) : y \succeq_i x_i\}$ is a polytope, so it suffices to consider only the vertices.
The dual of this linear program is
\begin{align*}
    \min 
        & \textstyle \sum_{i \in T_k} \sum_{y \in C^\alpha_i(p) : y \succeq_i x_i} \beta_{iy} \cdot p^{(k)} \cdot (y - x_i) \\
    \mathrm{s.t.} 
        & \textstyle \sum_{i \in T_k} \sum_{y \in C^\alpha_i(p) : y \succeq_i x_i} \beta_{iy} \cdot (y_j - x_{ij}) = 0 \,\forall\, j \not\in S_k \\
        & \beta_{iy} \geq 0 \,\forall\, i \in T_k, y \succeq_i x_i
\end{align*}
Note that 0 is a feasible solution to the dual, so by strong duality, if the primal is infeasible, then there exists a feasible dual $(\beta_{iy})_{iy}$ with negative objective value, contradicting the fact that $(x,p,\alpha)$ satisfies the aggregate cheapest bundle property.

Now, for all $k \in [d]$, let
\begin{align*}
    q^{(k)} &\coloneqq p^{(k)} + r^{(k)}\\
    \gamma^{(k)}_i &\coloneqq \max\{q^{(k)} \cdot (x_i - \omega_i), 0\} 
\end{align*}
Since the simplicity of $p$ implies that $p^{(k)}_j = 0$ for all $j \not\in S_k$, we have that
\begin{equation}\label{equation:q-decomposition-item-wise}
    q^{(k)}_j = p^{(k)}_j \cdot \1(j \in S_k) + r^{(k)}_j \cdot \1(j \not\in S_k)
\end{equation}
for all $j \in [m]$ and $k \in [d]$.
Moreover, the first non-zero entry in $q_j$ coincides with the first non-zero entry in $p_j$.
To see why, note that $S_1 \supseteq \dots \supseteq S_{d+1}$.
Thus, Equation~\ref{equation:q-decomposition-item-wise} implies that for $j \in S_k \setminus S_{k+1}$,
\begin{equation}\label{equation:first-non-zero-p-q-same}
    q^{(\ell)}_j = p^{(\ell)}_j \cdot \1(\ell = k) + r^{(\ell)}_j \cdot \1(\ell \not\in [k])
\end{equation}
Another important consequence of Equation~\ref{equation:q-decomposition-item-wise} is that for all $y_i$ such that $y_{ij} = 0$ for all $j \not\in S_k$,
\begin{equation}\label{equation:q-p-coincide}
    q^{(k)} \cdot y_i 
        = \sum_{j \in S_{k}} p^{(k)}_j \cdot y_{ij} + \sum_{j \not\in S_{k}} q^{(k)}_j \cdot y_{ij} 
        = \sum_{j \in S_{k}} p^{(k)}_j \cdot y_{ij} 
        = p^{(k)} \cdot y_i 
\end{equation}

We now argue that $C^\gamma_i(q) = C^\alpha_i(p)$ for all agents $i$.
Let
\[
    k_i \coloneqq \min \{k \in [d] : p^{(k)} \cdot \omega_i + \alpha^{(k)}_i > 0\} \cup \{d\},
\]
and note that
\begin{align*}
    C^\alpha_i(p) 
        &= \{y_i \in \Delta^n_i : p^{(k)} \cdot y_i \leq p^{(k)} \cdot \omega_i + \alpha^{(k)}_i \,\forall\, k \in [k_i]\} \tag{definition} \\
        &= \textstyle \{y_i \in \Delta^n_- : y_{ij} = 0 \,\forall\, j \not\in S_{k_i} \wedge \sum_{j \in S_{k_i}} p^{(k_i)}_j \cdot y_{ij} \leq \sum_{j \in S_{k_i}} p^{(k_i)}_j \cdot \omega_{ij} + \alpha^{(k_i)}_i\} \tag{$p \in \RR^{m \times d}_+$}
\end{align*}
The last equality follows from the fact that $p^{(k)} \cdot y_i \leq p^{(k)} \cdot \omega_i + \alpha^{(k)}_i \leq 0$ implies that $y_{ij} = \omega_{ij} = 0$ for all $k \in [d]$ and $j \not\in S_k$ (which implicitly uses the non-negativity of $p$).
It follows that Equation~\ref{equation:q-p-coincide} holds for $x_i$ and $\omega_i$ for all $k \in [k_i]$ since both lie in $C^\alpha_i(p)$ and $S_k \supseteq S_{k_i}$ for all $k \in [k_i]$, so
\[
    \gamma^{(k)}_i = \max\{q^{(k)} \cdot (x_i - \omega_i), 0\} = \max\{p^{(k)} \cdot (x_i - \omega_i),0\} = \alpha^{(k)}_i
\]
and
\[
    q^{(k)} \cdot \omega_i + \gamma^{(k)}_i = p^{(k)} \cdot \omega_i + \alpha^{(k)}_i
\]
for all $k \in [k_i]$.
At this point, it is not hard to see that $C^\alpha_i(p) \subseteq C^\gamma_i(q)$ using Equation~\ref{equation:q-p-coincide} and the alternative definition of $C^\alpha_i(p)$.
Moreover, note that the first currency that agent $i$ has positive wealth in under $q$ is exactly $k_i$, the first currency that she has positive wealth in under $p$, by Equation~\ref{equation:first-non-zero-p-q-same}.

To prove that $C^\alpha_i(p) \supseteq C^\gamma_i(q)$, it suffices to show that $y_{ij} = 0$ for all $y_i \in C^\gamma_i(q)$ and $j \not\in S_{k_i}$ as well.
Given $y_i \in C^\gamma_i(q)$, let
\[
    \ell \coloneqq \max\{k \in [d] : y_{ij} = 0 \,\forall\, j \not\in S_k\}
\]
and suppose by way of contradiction that $\ell \in [k_i - 1]$.
The maximality of $\ell$ implies that $y_{ij} = 0$ for all $j \not\in S_\ell$ and $y_{ij} > 0$ for some $j \in S_\ell \setminus S_{\ell+1}$.
By Equation~\ref{equation:q-p-coincide},
\[
    q^{(\ell)} \cdot y_i = \sum_{j \in S_\ell \setminus S_{\ell+1}} p^{(\ell)}_j \cdot y_{ij} > 0
\]
contradicting the fact that $q^{(\ell)} \cdot y_i \leq q^{(\ell)} \cdot \omega_i + \gamma^{(\ell)}_i = p^{(\ell)} \cdot \omega_i + \alpha^{(\ell)}_i = 0$.
Thus, $C^\gamma_i(q) = C^\alpha_i(p)$ for all agents $i$, and $(x, q, \gamma)$ constitutes an EDE.

It remains to show that $(x, q, \gamma)$ satisfies the strong cheapest bundle property.
Let $y_i \in C^\gamma_i(q) = C^\alpha_i(p)$ such that $y_i \succeq_i x_i$.
Since $x_i \in C^\gamma_i(q) = C^\alpha_i(p)$ as well, we have that $x_{ij} = y_{ij} = 0$ for all $j \not\in S_{k_i}$, so by Equation~\ref{equation:q-p-coincide}, the fact that $S_k \subseteq S_{k_i}$ for all $k \in [k_i]$, and the weak cheapest bundle property of $(x,p,\alpha)$, we have
\[
    q^{(k)} \cdot y_i = p^{(k)} \cdot y_i = p^{(k)} \cdot x_i = q^{(k)} \cdot x_i
\]
for all $k \in [k_i]$.
For $k \not\in [k_i]$, simply note that $i \in T_k$, so $q^{(k)} \cdot y_i \geq q^{(k)} \cdot x_i$ by definition of $q^{(k)}$.
\end{proof}

% In the interest of anonymization, please do not include acknowledgements in your submission.
%
%\begin{acks}
%
%	The authors would like to thank Dr. Maura Turolla of Telecom
%	Italia for providing specifications about the application scenario.
%
%	The work is supported by the \grantsponsor{GS501100001809}{National
%		Natural Science Foundation of
%		China}{http://dx.doi.org/10.13039/501100001809} under Grant
%	No.:~\grantnum{GS501100001809}{61273304\_a}
%	and~\grantnum[http://www.nnsf.cn/youngscientsts]{GS501100001809}{Young
%		Scientsts' Support Program}.
%
%
%\end{acks}

% Bibliography
\bibliographystyle{alpha}
\bibliography{references}

@Article{HyllandZ1979,
journal={Journal of Political Economy},
author={Hylland, Aanund and Zeckhauser, Richard},
title={The Efficient Allocation of Individuals to Positions},
year={1979},
month={April},
pages={293-314},
volume={87},
number={2},
doi={10.1086/260757},
url={https://ideas.repec.org/a/ucp/jpolec/v87y1979i2p293-314.html},
}

@Inbook{Mas-Colell1992,
author="Mas-Colell, Andreu",
editor="Majumdar, Mukul",
title="Equilibrium Theory with Possibly Satiated Preferences",
bookTitle="Equilibrium and Dynamics: Essays in Honour of David Gale",
year="1992",
publisher="Palgrave Macmillan UK",
address="London",
pages="201--213",
isbn="978-1-349-11696-6",
doi="10.1007/978-1-349-11696-6_9",
url="https://doi.org/10.1007/978-1-349-11696-6_9"
}

@article{AbrahamSW1996,
 ISSN = {09382259, 14320479},
 URL = {http://www.jstor.org/stable/25054959},
 author = {H. Abraham and D. R. Smart and J. Whittaker},
 journal = {Economic Theory},
 number = {1},
 pages = {177--182},
 publisher = {Springer},
 title = {Existence of Competitive Equilibrium When Incomes May Be Zero},
 urldate = {2026-02-03},
 volume = {8},
 year = {1996}
}

@article{EcheniqueMZ2021Constrained,
Author = {Echenique, Federico and Miralles, Antonio and Zhang, Jun},
Title = {Constrained Pseudo-Market Equilibrium},
Journal = {American Economic Review},
Volume = {111},
Number = {11},
Year = {2021},
Month = {November},
Pages = {3699–3732},
DOI = {10.1257/aer.20201769},
URL = {https://www.aeaweb.org/articles?id=10.1257/aer.20201769}}

@article{EcheniqueMZ23,
title = {Balanced equilibrium in pseudo-markets with endowments},
journal = {Games and Economic Behavior},
volume = {141},
pages = {428-443},
year = {2023},
issn = {0899-8256},
doi = {https://doi.org/10.1016/j.geb.2023.07.002},
url = {https://www.sciencedirect.com/science/article/pii/S0899825623000957},
author = {Federico Echenique and Antonio Miralles and Jun Zhang}
}

@article{Le2017,
author = {Le, Phuong},
title = {Competitive equilibrium in the random assignment problem},
journal = {International Journal of Economic Theory},
volume = {13},
number = {4},
pages = {369-385},
doi = {https://doi.org/10.1111/ijet.12135},
url = {https://onlinelibrary.wiley.com/doi/abs/10.1111/ijet.12135},
eprint = {https://onlinelibrary.wiley.com/doi/pdf/10.1111/ijet.12135},
year = {2017}
}

@article{GargTV2024,
	author = {Garg, Jugal and Tr{\"o}bst, Thorben and Vazirani, Vijay},
	date = {2024/08/12},
	doi = {10.1007/s10458-024-09670-9},
	id = {Garg2024},
	isbn = {1573-7454},
	journal = {Autonomous Agents and Multi-Agent Systems},
	number = {2},
	pages = {40},
	title = {One-sided matching markets with endowments: equilibria and algorithms},
	url = {https://doi.org/10.1007/s10458-024-09670-9},
	volume = {38},
	year = {2024}}

@article{BudishCKM2013,
Author = {Budish, Eric and Che, Yeon-Koo and Kojima, Fuhito and Milgrom, Paul},
Title = {Designing Random Allocation Mechanisms: Theory and Applications},
Journal = {American Economic Review},
Volume = {103},
Number = {2},
Year = {2013},
Month = {April},
Pages = {585–623},
DOI = {10.1257/aer.103.2.585},
URL = {https://www.aeaweb.org/articles?id=10.1257/aer.103.2.585}}

@article{Eisenberg1961,
 ISSN = {00251909, 15265501},
 URL = {http://www.jstor.org/stable/2627055},
 author = {E. Eisenberg},
 journal = {Management Science},
 number = {4},
 pages = {337--350},
 publisher = {INFORMS},
 title = {Aggregation of Utility Functions},
 urldate = {2026-02-04},
 volume = {7},
 year = {1961}
}

@article{Varian1974,
title = {Equity, envy, and efficiency},
journal = {Journal of Economic Theory},
volume = {9},
number = {1},
pages = {63-91},
year = {1974},
issn = {0022-0531},
doi = {https://doi.org/10.1016/0022-0531(74)90075-1},
url = {https://www.sciencedirect.com/science/article/pii/0022053174900751},
author = {Hal R Varian}
}

@inproceedings{NguyenT2024,
author = {Nguyen, Th\`{a}nh and Teytelboym, Alexander},
title = {Equilibrium in Pseudomarkets},
year = {2024},
isbn = {9798400707049},
publisher = {Association for Computing Machinery},
address = {New York, NY, USA},
url = {https://doi.org/10.1145/3670865.3673504},
doi = {10.1145/3670865.3673504},
booktitle = {Proceedings of the 25th ACM Conference on Economics and Computation},
pages = {1289},
numpages = {1},
location = {New Haven, CT, USA},
series = {EC '24}
}

@article{ShapleyS1974,
title = {On cores and indivisibility},
journal = {Journal of Mathematical Economics},
volume = {1},
number = {1},
pages = {23-37},
year = {1974},
issn = {0304-4068},
doi = {https://doi.org/10.1016/0304-4068(74)90033-0},
url = {https://www.sciencedirect.com/science/article/pii/0304406874900330},
author = {Lloyd Shapley and Herbert Scarf}
}

@article{RothP1977,
title = {Weak versus strong domination in a market with indivisible goods},
journal = {Journal of Mathematical Economics},
volume = {4},
number = {2},
pages = {131-137},
year = {1977},
issn = {0304-4068},
doi = {https://doi.org/10.1016/0304-4068(77)90004-0},
url = {https://www.sciencedirect.com/science/article/pii/0304406877900040},
author = {Alvin E. Roth and Andrew Postlewaite}
}

@article{Roth1982,
title = {Incentive compatibility in a market with indivisible goods},
journal = {Economics Letters},
volume = {9},
number = {2},
pages = {127-132},
year = {1982},
issn = {0165-1765},
doi = {https://doi.org/10.1016/0165-1765(82)90003-9},
url = {https://www.sciencedirect.com/science/article/pii/0165176582900039},
author = {Alvin E. Roth}
}

@article{HeMPY2018,
Author = {He, Yinghua and Miralles, Antonio and Pycia, Marek and Yan, Jianye},
Title = {A Pseudo-Market Approach to Allocation with Priorities},
Journal = {American Economic Journal: Microeconomics},
Volume = {10},
Number = {3},
Year = {2018},
Month = {August},
Pages = {272–314},
DOI = {10.1257/mic.20150259},
URL = {https://www.aeaweb.org/articles?id=10.1257/mic.20150259}}

@article{bergstrom1976,
title = {How to discard ‘free disposability’ - at no cost},
journal = {Journal of Mathematical Economics},
volume = {3},
number = {2},
pages = {131-134},
year = {1976},
issn = {0304-4068},
doi = {https://doi.org/10.1016/0304-4068(76)90021-5},
url = {https://www.sciencedirect.com/science/article/pii/0304406876900215},
author = {Theodore C. Bergstrom}
}

@article{PolemarchakisS1993,
title = {Competitive equilibria without free disposal or nonsatiation},
journal = {Journal of Mathematical Economics},
volume = {22},
number = {1},
pages = {85-99},
year = {1993},
issn = {0304-4068},
doi = {https://doi.org/10.1016/0304-4068(93)90032-G},
url = {https://www.sciencedirect.com/science/article/pii/030440689390032G},
author = {H.M. Polemarchakis and P. Siconolfi}
}

@article{GulPZ2024,
author = {Gul, Faruk and Pesendorfer, Wolfgang and Zhang, Mu},
title = {Efficient Allocation of Indivisible Goods in Pseudomarkets with Constraints},
journal = {Journal of Political Economy},
volume = {132},
number = {11},
pages = {3708-3736},
year = {2024},
doi = {10.1086/730561},

URL = { 
    
        https://doi.org/10.1086/730561
    
    

},
eprint = { 
    
        https://doi.org/10.1086/730561
    
    

}
}

@misc{GulP2020,
      title={Lindahl Equilibrium as a Collective Choice Rule}, 
      author={Faruk Gul and Wolfgang Pesendorfer},
      year={2020},
      eprint={2008.09932},
      archivePrefix={arXiv},
      primaryClass={econ.TH},
      url={https://arxiv.org/abs/2008.09932}, 
}

@article{ArrowD1954,
 ISSN = {00129682, 14680262},
 URL = {http://www.jstor.org/stable/1907353},
 author = {Kenneth J. Arrow and Gerard Debreu},
 journal = {Econometrica},
 number = {3},
 pages = {265--290},
 publisher = {[Wiley, Econometric Society]},
 title = {Existence of an Equilibrium for a Competitive Economy},
 urldate = {2026-02-04},
 volume = {22},
 year = {1954}
}

@article{Debreu1962,
 ISSN = {00206598, 14682354},
 URL = {http://www.jstor.org/stable/2525394},
 author = {Gerard Debreu},
 journal = {International Economic Review},
 number = {3},
 pages = {257--273},
 publisher = {[Economics Department of the University of Pennsylvania, Wiley, Institute of Social and Economic Research, Osaka University]},
 title = {New Concepts and Techniques for Equilibrium Analysis},
 urldate = {2026-02-04},
 volume = {3},
 year = {1962}
}

@article{McKenzie1959,
 ISSN = {00129682, 14680262},
 URL = {http://www.jstor.org/stable/1907777},
 author = {Lionel W. McKenzie},
 journal = {Econometrica},
 number = {1},
 pages = {54--71},
 publisher = {[Wiley, Econometric Society]},
 title = {On the Existence of General Equilibrium for a Competitive Market},
 urldate = {2026-02-04},
 volume = {27},
 year = {1959}
}

@article{Maxfield1997,
title = {General equilibrium and the theory of directed graphs},
journal = {Journal of Mathematical Economics},
volume = {27},
number = {1},
pages = {23-51},
year = {1997},
issn = {0304-4068},
doi = {https://doi.org/10.1016/0304-4068(95)00763-6},
url = {https://www.sciencedirect.com/science/article/pii/0304406895007636},
author = {Robert R. Maxfield}
}

@Inbook{Eaves1985,
author="Curtis Eaves, B.",
editor="Manne, Alan S.",
title="Finite solution of pure trade markets with Cobb-Douglas utilities",
bookTitle="Economic Equilibrium: Model Formulation and Solution",
year="1985",
publisher="Springer Berlin Heidelberg",
address="Berlin, Heidelberg",
pages="226--239",
isbn="978-3-642-00917-4",
doi="10.1007/BFb0121035",
url="https://doi.org/10.1007/BFb0121035"
}

@article{Moore1975,
 ISSN = {00206598, 14682354},
 URL = {http://www.jstor.org/stable/2525812},
 author = {James C. Moore},
 journal = {International Economic Review},
 number = {2},
 pages = {267--300},
 publisher = {[Economics Department of the University of Pennsylvania, Wiley, Institute of Social and Economic Research, Osaka University]},
 title = {The Existence of "Compensated Equilibrium" and the Structure of the Pareto Efficiency Frontier},
 urldate = {2026-02-04},
 volume = {16},
 year = {1975}
}

@article{DebreuS1963,
 ISSN = {00206598, 14682354},
 URL = {http://www.jstor.org/stable/2525306},
 author = {Gerard Debreu and Herbert Scarf},
 journal = {International Economic Review},
 number = {3},
 pages = {235--246},
 publisher = {[Economics Department of the University of Pennsylvania, Wiley, Institute of Social and Economic Research, Osaka University]},
 title = {A Limit Theorem on the Core of an Economy},
 urldate = {2026-02-04},
 volume = {4},
 year = {1963}
}

@book{Edgeworth1881,
  title={Mathematical Psychics: An Essay on the Application of Mathematics to the Moral Sciences},
  author={Edgeworth, Francis Ysidro},
  year={1881},
  publisher={C. Kegan Paul \& Co.},
  address={London}
}

@book{ArrowH1971,
  author    = {Arrow, Kenneth J. and Hahn, Frank H.},
  title     = {General Competitive Analysis},
  year      = {1971},
  series    = {Mathematical Economics Texts},
  publisher = {Holden-Day},
  address   = {San Francisco}
}

@book{VonNeumannM1944,
  author    = {von Neumann, John and Morgenstern, Oskar},
  title     = {Theory of Games and Economic Behavior},
  year      = {1944},
  publisher = {Princeton University Press},
  address   = {Princeton, NJ}
}

@incollection{Gillies1959,
  title = {Solutions to general non-zero-sum games},
  booktitle = {Contributions to the Theory of Games},
  volume = {4},
  author = {Gillies, Donald B.},
  editor = {Albert William Tucker and Robert Duncan Luce},
  year = {1959},
  pages = {47--83},
  publisher = {Princeton University Press},
  address = {Princeton, NJ},
  series = {Annals of Mathematics Studies}
}

@article{Aumann1961,
 ISSN = {00029947},
 URL = {http://www.jstor.org/stable/1993348},
 author = {Robert J. Aumann},
 journal = {Transactions of the American Mathematical Society},
 number = {3},
 pages = {539--552},
 publisher = {American Mathematical Society},
 title = {The Core of a Cooperative Game Without Side Payments},
 urldate = {2026-02-04},
 volume = {98},
 year = {1961}
}

@article{Peleg1963,
 ISSN = {00029947},
 URL = {http://www.jstor.org/stable/1993770},
 author = {Bezalel Peleg},
 journal = {Transactions of the American Mathematical Society},
 number = {2},
 pages = {280--292},
 publisher = {American Mathematical Society},
 title = {Solutions to Cooperative Games Without Side Payments},
 urldate = {2026-02-04},
 volume = {106},
 year = {1963}
}

@article{Scarf1967,
 ISSN = {00129682, 14680262},
 URL = {http://www.jstor.org/stable/1909383},
 author = {Herbert E. Scarf},
 journal = {Econometrica},
 number = {1},
 pages = {50--69},
 publisher = {[Wiley, Econometric Society]},
 title = {The Core of an N Person Game},
 urldate = {2026-02-04},
 volume = {35},
 year = {1967}
}

@incollection{Anderson1992,
title = {Chapter 14 The core in perfectly competitive economies},
series = {Handbook of Game Theory with Economic Applications},
publisher = {Elsevier},
volume = {1},
pages = {413-457},
year = {1992},
issn = {1574-0005},
doi = {https://doi.org/10.1016/S1574-0005(05)80017-7},
url = {https://www.sciencedirect.com/science/article/pii/S1574000505800177},
author = {Robert M. Anderson}
}

@article{Aumann1964,
 ISSN = {00129682, 14680262},
 URL = {http://www.jstor.org/stable/1913732},
 author = {Robert J. Aumann},
 journal = {Econometrica},
 number = {1/2},
 pages = {39--50},
 publisher = {[Wiley, Econometric Society]},
 title = {Markets with a Continuum of Traders},
 urldate = {2026-02-04},
 volume = {32},
 year = {1964}
}

@techreport{McLennan2018,
author = {Andrew McLennan},
title = {Efficient disposal equilibria of pseudomarkets},
year = {2018},
type = {Mimeo},
institution = {University of Queensland}
}

@article{Konovalov2005,
 ISSN = {09382259, 14320479},
 URL = {http://www.jstor.org/stable/25055907},
 author = {Alexander Konovalov},
 journal = {Economic Theory},
 number = {3},
 pages = {711--719},
 publisher = {Springer},
 title = {The Core of an Economy with Satiation},
 urldate = {2026-02-06},
 volume = {25},
 year = {2005}
}

@article{DrezeM1980,
title = {Optimality properties of rationing schemes},
journal = {Journal of Economic Theory},
volume = {23},
number = {2},
pages = {131-149},
year = {1980},
issn = {0022-0531},
doi = {https://doi.org/10.1016/0022-0531(80)90001-0},
url = {https://www.sciencedirect.com/science/article/pii/0022053180900010},
author = {Jacques H Drèze and Heinz Müller}
}

@article{AumannD1986,
 ISSN = {00129682, 14680262},
 URL = {http://www.jstor.org/stable/1914300},
 author = {Robert J. Aumann and Jacques H. Drèze},
 journal = {Econometrica},
 number = {6},
 pages = {1271--1318},
 publisher = {[Wiley, Econometric Society]},
 title = {Values of Markets with Satiation or Fixed Prices},
 urldate = {2026-02-06},
 volume = {54},
 year = {1986}
}

@article{Makarov1981,
title = {Some results on general assumptions about the existence of economic equilibrium},
journal = {Journal of Mathematical Economics},
volume = {8},
number = {1},
pages = {87-99},
year = {1981},
issn = {0304-4068},
doi = {https://doi.org/10.1016/0304-4068(81)90014-8},
url = {https://www.sciencedirect.com/science/article/pii/0304406881900148},
author = {V.L. Makarov}
}

@Inbook{MiyazakiT2012,
author="Miyazaki, Kentaro
and Takekuma, Shin-Ichi",
editor="Kusuoka, Shigeo
and Maruyama, Toru",
title="On the equivalence between the rejective core and the dividend equilibrium: a note",
bookTitle="Advances in Mathematical Economics Volume 16",
year="2012",
publisher="Springer Japan",
address="Tokyo",
pages="111--117",
isbn="978-4-431-54114-1",
doi="10.1007/978-4-431-54114-1_5",
url="https://doi.org/10.1007/978-4-431-54114-1_5"
}

@article{MurakamiU2017,
  author = {Murakami, Hiromi and Urai, Ken},
  title = {Replica core limit theorem for economies with satiation},
  journal = {Economic Theory Bulletin},
  year = {2017},
  volume = {5},
  pages = {259--270},
  month = {Oct},
  doi = {10.1007/s40505-017-0119-2},
  url = {https://link.springer.com/article/10.1007/s40505-017-0119-2}
}

% Appendix
\appendix

\section{Existence: Strict Endowments}

In this section, we demonstrate the existence of a 1-dimensional LDE satisfying the strong cheapest bundle property if endowments are strictly positive.
We will use this result later to prove that an LDE satisfying this property exists even without this condition.
The main result of \cite{McLennan2018} ``qualitatively'' implies our result, but it is potentially interesting that there exists an equilibrium supported on VCG prices.

\begin{theorem}\label{theorem:VCG-LDE}
If $\omega \gg 0$, then there exists $\lambda \in \RR_+^n$, $\alpha \in \RR_+$, and $x \in \arg\max \{\sum_i \lambda_i u_i \cdot y_i : y \in \mathcal{M}\}$ such that for all agents $i$,
\[
    u_i \cdot x_i = \max_{y_i \in \Delta_-^n} \{u_i \cdot y_i : p(\lambda) \cdot y_i \leq p(\lambda) \cdot \omega_i + \alpha\}
\]
where $p_j(\lambda)$ denotes the VCG price of good $j$ with respect to utilities $(\lambda_i u_i)_i$.
\end{theorem}

\begin{lemma}
If $\lambda \in \RR_+^n$ and $x \in \arg\max \{\sum_i \lambda_i u_i \cdot y_i : y \in \mathcal{M}\}$, then 
\[
    u_i \cdot y_i \geq u_i \cdot x_i \implies p(\lambda) \cdot y_i \geq p(\lambda) \cdot x_i.
\]
\end{lemma}

\begin{proof}
The dual linear program is
\begin{align*}
    \min & \textstyle \sum_i \alpha_i + \sum_j \beta_j \\
    \text{s.t. } 
        & \alpha_i + \beta_j \geq \lambda_i u_{ij} \,\forall\, i, j \\
        & \alpha, \beta \geq 0
\end{align*}
A property of VCG prices is that $\alpha_i = \lambda_i u_{ij} - p_j(\lambda)$ for any $j$ such that $x_{ij} > 0$ and $\beta_j = p_j(\lambda)$ is an optimal solution to the dual.
Thus, for any $y_i$ such that $u_i \cdot y_i \geq u_i \cdot x_i$, we have that
\[
    p(\lambda) \cdot x_i = \lambda_i u_i \cdot x_i - \alpha_i \leq \lambda_i u_i \cdot y_i - \alpha_i \leq p(\lambda) \cdot y_i
\]
\end{proof}

\subsection{Proof of Theorem~\ref{theorem:VCG-LDE}}

\begin{lemma}\label{lemma:VCG-price-continuous}
$p(\lambda)$ is continuous in $\lambda$.
\end{lemma}

\begin{proof}
Recall that 
\[
    \textstyle p_\ell(\lambda) = \max \{\sum_i \lambda_i u_i \cdot x_i : \sum_j x_{ij} \leq 1 \,\forall\, i, \sum_i x_{ij} \leq 1 \,\forall\, j \not= \ell, \sum_i x_{i\ell} \leq 2\} - \max_{x \in \mathcal{M}} \sum_i \lambda_i u_i \cdot x_i.
\]
The optimal objective values are continuous in $\lambda$ by the maximum theorem, so $p_\ell(\lambda)$ is continuous in $\lambda$.
\end{proof}

\begin{lemma}\label{lemma:optimal-consumption-continuous-1D-LDE}
Given $\lambda$ and $\alpha$, define $C_i(\lambda, \alpha) \coloneqq \{y_i \in \Delta_-^n : p(\lambda) \cdot y_i \leq p(\lambda) \cdot \omega_i + \alpha\}$.
Suppose $\lambda^t \to \lambda$ and $\alpha^t \to \alpha$.
If $p(\lambda) \not= 0$, then
\[
    \textstyle \lim_{t \to \infty} \max \{u_i \cdot y_i : y_i \in C_i(\lambda^t, \alpha^t)\} = \max \{u_i \cdot : y_i \in C_i(\lambda, \alpha)\}
\]
\end{lemma}

\begin{proof}
For each $t$, let $x^t_i \in \arg\max \{u_i \cdot y_i : y_i \in C_i(\lambda^t, \alpha^t)\}$, and suppose $x^t_i \to x_i$ as $t \to \infty$ (after passing to a subsequence).
It is clear that $x_i \in C_i(\lambda, s)$ (since $p$ is continuous by Lemma~\ref{lemma:VCG-price-continuous}).
We now show that $x_i$ is in fact a maximizer.
Let $x_i^* \in \arg\max \{u_i \cdot y_i: y_i \in C_i(\lambda, s)\}$, and suppose by way of contradiction that $u_i \cdot x_i < u_i \cdot x_i^*$, so $x^*_i \not= 0$.
Let $\zeta$ denote the gap.
Define $r \coloneqq \min \{x^*_{ij} : j \in \supp{x^*_i}\} \cup \{\zeta / 2\}$ and $x'_i$ item-wise as $x'_{ij} \coloneqq (x^*_{ij} - r)^+$.
Clearly, $u_i \cdot x'_i \geq u_i \cdot x^*_i - \zeta / 2$.
We now show that for sufficiently large $t$, $x'_i \in C_i(\lambda^t, \alpha^t)$.

Since $\lambda^t \to \lambda$, $\alpha^t \to \alpha$, and $p$ is continuous by Lemma~\ref{lemma:VCG-price-continuous}, we eventually have that 
\[
    \lVert p(\lambda) - p(\lambda^t)\rVert_\infty + \lvert \alpha^t\rvert \leq \gamma
\] 
where $\gamma \coloneqq p(\lambda) \cdot \omega_i / ((n + 1)/r + n)$, which is strictly positive since $p(\lambda) \not= 0$ and $\omega \gg 0$.
If $\sum_{j : x^*_{ij} > 0} p_j(\lambda^t) \leq \gamma (n + 1) / r$, then 
\[
     \textstyle p(\lambda^t) \cdot x'_i \leq p(\lambda^t) \cdot x^*_i \leq \sum_{j : x^*_{ij} > 0} p_j(\lambda^t) \leq \gamma (n + 1) / r = p(\lambda) \cdot \omega_i - n \gamma \leq p(\lambda^t) \cdot \omega_i
\]
where the first inequality follows from the definition of $x'_i$ and the fourth inequality follows from the fact that $\norm{p(\lambda) - p(\lambda^t)}_\infty \leq \gamma$ and $\omega_i \in [0,1]^n$, so $x'_i \in C_i(\lambda^t, \alpha^t)$ in this case.
Otherwise, $\sum_{j : x^*_{ij} > 0} p_j(\lambda^t) \geq \gamma (n + 1) / r$, in which case
\begin{align*}
    p(\lambda^t) \cdot x'_{i}
        &= \textstyle p(\lambda^t) \cdot x_{i}^* - r \sum_{j : x^*_{ij} > 0} p_j(\lambda^t) \tag{definition of $x'_i$} \\
        &\leq \textstyle p(\lambda) \cdot x_{i}^* - r \sum_{j : x^*_{ij} > 0} p_j(\lambda^t) + \gamma \tag{$\norm{p(\lambda) - p(\lambda^t)}_\infty \leq \gamma$; $x^*_i \in \Delta_-^n$} \\
        &\leq \textstyle p(\lambda) \cdot \omega_i + \alpha - r \sum_{j : x^*_{ij} > 0} p_j(\lambda^t) + \gamma \tag{$x^*_i \in C_i(\lambda, \alpha)$} \\
        &\leq \textstyle p(\lambda^t) \cdot \omega_i + \alpha^t - r \sum_{j : x^*_{ij} > 0} p_j(\lambda^t) + \gamma (n + 1) \tag{$\norm{p(\lambda) - p(\lambda^t)}_\infty \leq \gamma$; $\omega_i \in [0,1]^n$; $\abs{s-s^t} \leq \gamma$} \\
        &\leq \textstyle p(\lambda^t) \cdot \omega_i + \alpha^t
\end{align*}
so again, $x'_i \in C_i(\lambda^t, \alpha^t)$.
It follows that for sufficiently large $t$, $u_i \cdot x^t_i \geq u_i \cdot x'_i \geq u_i \cdot x^*_i - \zeta / 2$.
But since $x^t_i \to x_i$ as $t \to \infty$, it is also true that for sufficiently large $t$, $u_i \cdot x^t_i < u_i \cdot x_i + \zeta / 2 \leq u_i \cdot x^*_i - \zeta / 2$, a contradiction.
\end{proof}

Now, define the following quantities:
\begin{align*}
    x^\delta(\lambda) 
        &\coloneqq \textstyle  \arg\max \{ \sum_i \lambda_i u_i \cdot x_i + \delta \norm{x}_F^2 : x \in \mathcal{M}\} \\
    C_i^\delta(\lambda) 
        &\coloneqq \{y_i \in \Delta_-^n : p(\lambda) \cdot y_i \leq p(\lambda) \cdot \omega_i + \max_k (p(\lambda) \cdot x_k^\delta(\lambda) - p(\lambda) \cdot \omega_k)^+\}\\
    u_i^\delta(\lambda) 
        &\coloneqq \textstyle \max \{u_i \cdot y_i : y_i \in C_i^\delta(\lambda)\} \\
    \varphi_i^{\varepsilon, \delta}(\lambda)
        &\coloneqq \textstyle \left[\frac{\lambda_i + [u_i^\delta(\lambda) - u_i \cdot x^\delta_i(\lambda)]_0^1}{\eta + (1 - \eta) \sum_j p_j(\lambda)} \right]_\varepsilon^{\Bar{\lambda}}
\end{align*}

\begin{lemma}\label{lemma:xdelta-continuous}
For all $\delta > 0$, $x^\delta(\lambda)$ is continuous in $\lambda$.
\end{lemma}

\begin{proof}
$\mathcal{M}$ is compact and continuous in $\lambda$ (since it does not even vary with $\lambda$), so $x^\delta$ is upper-hemicontinuous in $\lambda$ by the maximum theorem.
Continuity follows from the uniqueness of the maximizer (the objective is strictly concave).
\end{proof}

\begin{lemma}
For all $\delta > 0$, $u_i^\delta(\lambda)$ is continuous in $\lambda \in [\varepsilon, \bar{\lambda}]^n$ if $\omega \gg 0$ and $p(\lambda) \not= 0$ for all $\lambda \in [\varepsilon, \bar{\lambda}]^n$.
\end{lemma}

\begin{proof}
Since $p$ and $x^\delta$ are continuous by Lemmas~\ref{lemma:VCG-price-continuous} and~\ref{lemma:xdelta-continuous}, the result follows from Lemma~\ref{lemma:optimal-consumption-continuous-1D-LDE}.
\end{proof}

\begin{lemma}
If there exists $\varepsilon > 0$ and $\lambda \in [\varepsilon, \bar{\lambda}]^n$ such that $p(\lambda) = 0$, then for all $x_i \in \arg\max \{ \sum_i \lambda_i u_i \cdot y_i : y \in \mathcal{M}\}$, we have that $u_i \cdot x_i = \max_j u_{ij}$ for all agents $i$.
\end{lemma}

\begin{proof}
Suppose by way of contradiction that there exists an agent $i$ such that $u_i \cdot x_i < \max_j u_{ij}$.
Let $j^*$ denote her favorite item.
By the Birkhoff-von-Neumann theorem, there exists a matching $M$ in the support of $x_i$ such that $u_{i, M(i)} < u_{ij^*}$.
It follows that $p_{j^*}(\lambda) \geq \lambda_i (u_{ij^*} - u_{i, M(i)}) > 0$, a contradiction.
\end{proof}

Now, suppose for all $\varepsilon > 0$ and $\lambda \in [\varepsilon, \bar{\lambda}]^n$, we have that $p(\lambda) \not = 0$.
By Brouwer's fixed point theorem, we have that for all $\delta > 0$, there exists $\lambda^{\varepsilon, \delta} \in [\varepsilon, \bar{\lambda}]^n$ such that $\lambda^{\varepsilon, \delta} = \varphi^{\varepsilon, \delta}(\lambda^{\varepsilon, \delta})$.
For each $\varepsilon > 0$, consider the limit of $\lambda^{\varepsilon, \delta}$ (after passing to a subsequence) as $\delta \to 0$.
Let $\lambda^\varepsilon$ denote the limit of $\lambda^{\varepsilon, \delta}$ and $x(\lambda^\varepsilon)$ the limit of $x^\delta(\lambda^{\varepsilon, \delta})$.

\begin{lemma}
$x(\lambda^\varepsilon) \in \arg\max \{\sum_i \lambda^\varepsilon u_i \cdot x_i : x \in \mathcal{M}\}$.
\end{lemma}

\begin{proof}
Let $x^* \in \arg\max \{\sum_i \lambda^\varepsilon_i u_i \cdot x_i: x \in \mathcal{M}\}$.
Suppose by way of contradiction that $\sum_i \lambda^\varepsilon_i u_i \cdot x_i(\lambda^\varepsilon) < \sum_i \lambda^\varepsilon_i u_i \cdot x_i^*$.
Let $\zeta$ denote the gap, and note that for sufficiently small $\delta > 0$, we have that
\[
    \textstyle \sum_i \lambda^{\varepsilon, \delta} u_i \cdot x^\delta_i(\lambda^{\varepsilon, \delta}) + \delta \norm{x^\delta(\lambda^{\varepsilon, \delta})}^2_F 
        < \sum_i \lambda^{\varepsilon} u_i \cdot x_i(\lambda^\varepsilon) + \zeta / 2 
        = \sum_i \lambda^\varepsilon u_i \cdot x_i^* - \zeta / 2 < \sum_i \lambda^{\varepsilon, \delta} u_i \cdot x^*_i
\]
which contradicts the definition of $x^\delta(\lambda^{\varepsilon, \delta})$.
\end{proof}

\begin{lemma}\label{lemma:utility-favorite-bundle-continuous-delta}
Let $C^0_i(\lambda^\varepsilon) = \{y_i \in \Delta_-^n : p(\lambda^\varepsilon) \cdot y_i \leq p(\lambda^\varepsilon) \cdot \omega_i + \max_k (p(\lambda^\varepsilon) \cdot x_k(\lambda^\varepsilon) - p(\lambda^\varepsilon) \cdot \omega_k)^+\}$.
\[
    \textstyle \lim_{\delta \to 0} u_i^\delta(\lambda^{\varepsilon, \delta}) = \max \{u_i \cdot y_i : y_i \in C^0_i(\lambda^\varepsilon)\}
\]
\end{lemma}

\begin{proof}
Since $\lambda^{\varepsilon, \delta} \to \lambda^\varepsilon$ and $x^\delta(\lambda^{\varepsilon, \delta}) \to x(\lambda^\varepsilon)$ by definition, we have that
\[
    \max_k (p(\lambda^{\varepsilon, \delta}) \cdot x_k^\delta(\lambda^{\varepsilon, \delta}) - p(\lambda^{\varepsilon, \delta}) \cdot \omega_k)^+ \to \max_k (p(\lambda^\varepsilon) \cdot x_k(\lambda^\varepsilon) - p(\lambda^\varepsilon) \cdot \omega_k)^+,
\]
as $p$ is continuous by Lemma~\ref{lemma:VCG-price-continuous}, so the result follows by Lemma~\ref{lemma:optimal-consumption-continuous-1D-LDE}.
\end{proof}

Define
\[
    \textstyle u_i(\lambda) \coloneqq \max_{y_i \in C^0_i(\lambda)} u_i \cdot y_i
\]
By Lemma \ref{lemma:utility-favorite-bundle-continuous-delta}, for all $\varepsilon > 0$, we have that
\[
    \textstyle \lambda^\varepsilon_i = \lim_{\delta \to 0} \lambda^{\varepsilon, \delta}_i = \lim_{\delta \to 0} \varphi^{\varepsilon, \delta}_i(\lambda^{\varepsilon, \delta}_i) =  \left[\frac{\lambda^\varepsilon + [u_i(\lambda^\varepsilon) - u_i \cdot x_i(\lambda^\varepsilon)]_0^1}{\eta + (1 - \eta) \sum_j p_j(\lambda^\varepsilon)} \right]_\varepsilon^{\Bar{\lambda}}
\]
Note that $x(\lambda^\varepsilon)$ allocates everything, so $\sum_i p(\lambda^\varepsilon) \cdot (x_i(\lambda^\varepsilon) - \omega_i) = 0$.
It follows that
\[
    \max_i (p(\lambda^\varepsilon) \cdot x_i(\lambda^\varepsilon) - p(\lambda^\varepsilon) \cdot \omega_i)^+ = \max_i p(\lambda^\varepsilon) \cdot (x_i(\lambda^\varepsilon) - \omega_i).
\]
Now, fix $\varepsilon > 0$, and consider the agent $i(\varepsilon) \in \arg\max_i p(\lambda^\varepsilon) \cdot (x_i(\lambda^\varepsilon) - \omega_i)$.

\begin{lemma}\label{lemma:greatest-overspender-consumes-optimally}
If $x(\lambda^\varepsilon) \in \arg\max \{\sum_i \lambda_i^\varepsilon u_i \cdot x_i : x \in \mathcal{M}\}$, then $u_{i(\varepsilon)}(\lambda^\varepsilon) = u_{i(\varepsilon)} \cdot x_{i(\varepsilon)}(\lambda^\varepsilon)$.
\end{lemma}

\begin{proof}
Note that $C^0_{i(\varepsilon)}(\lambda^\varepsilon) = \{y_{i(\varepsilon)} \in \Delta_-^n : p(\lambda^\varepsilon) \cdot (y_{i(\varepsilon)} - x_{i(\varepsilon)}(\lambda^\varepsilon)) \leq 0\}$.
The dual LP of $i(\varepsilon)$'s consumption problem is thus
\[
    \min_{a, b \geq 0} \{a + b p(\lambda^\varepsilon) \cdot x_{i(\varepsilon)}(\lambda^\varepsilon) : a + b p_j(\lambda^\varepsilon) \geq u_{i(\varepsilon), j}\}
\]
Now, consider the dual LP of the assignment problem with utilities $(\lambda^\varepsilon_i u_i)_i$.
We know that there exist $\alpha^*$ such that
\[
    \textstyle (\alpha^*, p(\lambda^\varepsilon)) \in \arg\min_{\alpha, \beta \geq 0} \{\sum_i \alpha_i + \sum_j \beta_j : \alpha_i + \beta_j \geq \lambda^\varepsilon_i u_{ij} \,\forall i, j\}
\]
We claim that $a \coloneqq \alpha^*_{i(\varepsilon)} / \lambda^\varepsilon_{i(\varepsilon)}$ and $b \coloneqq 1/\lambda^\varepsilon_{i(\varepsilon)}$ is an optimal solution to the dual LP of $i(\varepsilon)$'s consumption problem.
It is clearly a feasible dual solution.
Moreover, since $x(\lambda^\varepsilon)$ maximizes welfare with respect to $(\lambda^\varepsilon_i u_i)_i$, we have that $\sum_j x_{ij}(\lambda^\varepsilon) = 1$ and $x_{ij}(\lambda^\varepsilon) > 0 \implies \alpha_i^* + p_j(\lambda^\varepsilon) = \lambda^\varepsilon_i u_{ij}$.
Thus,
\begin{align*}
    a + b p(\lambda^\varepsilon) \cdot x_{i(\varepsilon)}(\lambda^\varepsilon) 
        &= \textstyle \frac{1}{\lambda_{i(\varepsilon)}^\varepsilon} (\alpha^*_{i(\varepsilon)} + p(\lambda^\varepsilon) \cdot x_{i(\varepsilon)}(\lambda^\varepsilon)) \\
        &= \textstyle \frac{1}{\lambda_{i(\varepsilon)}^\varepsilon} \sum_j x_{i(\varepsilon), j}(\lambda^\varepsilon) (\alpha_{i(\varepsilon)}^* + p_j(\lambda^\varepsilon)) \\
        &= \textstyle u_{i(\varepsilon)} \cdot x_{i(\varepsilon)}(\lambda^\varepsilon)
\end{align*}
Since $i(\varepsilon)$ has exactly enough money to consume $x_{i(\varepsilon)}(\lambda^\varepsilon)$, it is also feasible for $i(\varepsilon)$ to consume this lottery subject to her budget.
Since we have found a feasible primal solution and a feasible dual solution to $i(\varepsilon)$'s consumption problem that obtain the same objective value, it follows that $x_{i(\varepsilon)}(\lambda^\varepsilon)$ is an optimal consumption. 
\end{proof}

If $\lambda^\varepsilon_{i(\varepsilon)} > \varepsilon$, then
\[
    \lambda^\varepsilon_{i(\varepsilon)} 
        \leq \frac{\lambda^\varepsilon_{i(\varepsilon)} + [u_{i(\varepsilon)}(\lambda^\varepsilon) - u_{i(\varepsilon)} \cdot x_{i(\varepsilon)}(\lambda^\varepsilon)]_0^1}{\eta + (1 - \eta) \sum_j p_j(\lambda^\varepsilon)} 
        = \frac{\lambda^\varepsilon_{i(\varepsilon)}}{\eta + (1 - \eta) \sum_j p_j(\lambda^\varepsilon)}
\]
where the inequality follows from the fact that $\lambda^\varepsilon$ is a fixed point and the equality follows from Lemma \ref{lemma:greatest-overspender-consumes-optimally}. 
It follows that $\sum_j p_j(\lambda^\varepsilon) \leq 1$.
We claim that this inequality implies that $u_i(\lambda^\varepsilon) = u_i \cdot x_i(\lambda^\varepsilon)$ for all agents $i$.

\begin{lemma}
Let $\bar{\lambda} > 1/\min_{i,j,\ell : u_{ij} > u_{i\ell}}\{u_{ij} - u_{i\ell}\}$.
If $\sum_j p_j(\lambda^\varepsilon) \leq 1$ and $\lambda^\varepsilon_i = \bar{\lambda}$, then $u_i \cdot x_i(\lambda^\varepsilon) = \max_j u_{ij}$.
\end{lemma}

\begin{proof}
Let $j^* \in \arg\max_j u_{ij}$, and suppose by way of contradiction that $u_i \cdot x_i(\lambda^\varepsilon) < u_{ij^*}$.
By the Birkhoff-von-Neumann theorem, there exists a matching $M$ in the support of $x_i(\lambda^\varepsilon)$ such that $u_{i, M(i)} < u_{ij^*}$.
It follows that $p_{j^*}(\lambda^\varepsilon) \geq \lambda^\varepsilon_i (u_{ij^*} - u_{i, M(i)}) > 1$, a contradiction.
\end{proof}

\begin{lemma}
If $\sum_j p_j(\lambda^\varepsilon) \leq 1$ and $\varepsilon \leq \lambda^\varepsilon_i < \bar{\lambda}$, then $u_i \cdot x_i(\lambda^\varepsilon) = u_{i}(\lambda^\varepsilon)$.
\end{lemma}

\begin{proof}
If $u_i \cdot x_i(\lambda^\varepsilon) < u_{i}(\lambda^\varepsilon)$, then
\[
    \lambda^\varepsilon_i \geq \frac{\lambda^\varepsilon_{i} + [u_{i}(\lambda^\varepsilon) - u_{i} \cdot x_{i}(\lambda^\varepsilon)]_0^1}{\eta + (1 - \eta) \sum_j p_j(\lambda^\varepsilon)}
        > \frac{\lambda^\varepsilon_{i(\varepsilon)}}{\eta + (1 - \eta) \sum_j p_j(\lambda^\varepsilon)} \geq \lambda^\varepsilon_i
\] 
a contradiction. 
The first inequality follows from the fact that $\lambda^\varepsilon$ is a fixed point and $\varepsilon \leq \lambda^\varepsilon_i < \bar{\lambda}$.
The second inequality follows from $u_i \cdot x_i(\lambda^\varepsilon) < u_{i}(\lambda^\varepsilon)$.
The third inequality follows from $\sum_j p_j(\lambda^\varepsilon) \leq 1$.
\end{proof}

At this point, we have shown that if there exists $\varepsilon > 0$ such that $\lambda^\varepsilon_{i(\varepsilon)} > \varepsilon$ where $i(\varepsilon) \in \arg\max_i p(\lambda^\varepsilon) \cdot (x_i(\lambda^\varepsilon) - \omega_i)$, then Theorem \ref{theorem:VCG-LDE} holds with $\alpha = p(\lambda^\varepsilon) \cdot (x_{i(\varepsilon)}(\lambda^\varepsilon) - \omega_{i(\varepsilon)})$.
We now show how to obtain a contradiction if such a $\varepsilon$ does not exist.

Suppose that for all $\varepsilon > 0$, $\lambda_{i(\varepsilon)}^\varepsilon = \varepsilon$.
By the fact that $\lambda^\varepsilon$ is a fixed point and Lemma \ref{lemma:greatest-overspender-consumes-optimally},
\[
    \lambda_{i(\varepsilon)}^\varepsilon 
        \geq \frac{\lambda^\varepsilon_{i(\varepsilon)} + [u_{i(\varepsilon)}(\lambda^\varepsilon) - u_{i(\varepsilon)} \cdot x_{i(\varepsilon)}(\lambda^\varepsilon)]_0^1}{\eta + (1 - \eta) \sum_j p_j(\lambda^\varepsilon)} 
        = \frac{\lambda^\varepsilon_{i(\varepsilon)}}{\eta + (1 - \eta) \sum_j p_j(\lambda^\varepsilon)} 
\]
so $\sum_j p_j(\lambda^\varepsilon) \geq 1$.
Thus, if we were to take the limit of $\lambda^\varepsilon$ as $\varepsilon \to 0$ (after passing to a subsequence), then prices would not item-wise vanish.
Since $\lambda_{i(\varepsilon)}^\varepsilon = \varepsilon$ for all $\varepsilon > 0$, we know that
\[
    \lim_{\varepsilon \to 0} p(\lambda^\varepsilon) \cdot x_{i(\varepsilon)}(\lambda^\varepsilon) \leq \lim_{\varepsilon \to 0} \lambda_i^\varepsilon u_i \cdot x_{i(\varepsilon)}(\lambda^\varepsilon) \leq 0
\]
so
\[
    \max_i p(\lambda^0) \cdot (x_i(\lambda^0) - \omega_i) = \lim_{\varepsilon \to 0} \max_i p(\lambda^\varepsilon) \cdot (x_i(\lambda^\varepsilon) - \omega_i) < 0
\]
as $\sum_j p_j(\lambda^\varepsilon) \geq 1$ and $\omega_i \gg 0$.
But this is a contradiction since $\sum_i p(\lambda^0) \cdot (x_i(\lambda^0) - \omega_i) = 0$.

\section{Existence: General Case}

In this section, we formally prove that an LDE satisfying the strong cheapest bundle property exists.

\LDEexistence*

\subsection{Proof of Theorem~\ref{theorem:LDE-existence}}

Throughout the proof, define $x/y$ as follows for all $x, y \in \RR$.
\[
    x/y \coloneqq \begin{cases}
        x/y & y \not= 0 \\
        +\infty & x > 0 \wedge y = 0 \\
        -\infty & x < 0 \wedge y = 0 \\
        1 & x = y = 0
    \end{cases}
\]
Now, for all agents $i$, let $\omega^\varepsilon_i \coloneqq \frac{\varepsilon}{n} \cdot \1 + (1 - \varepsilon) \cdot \omega_i$.
By Theorem~\ref{theorem:VCG-LDE}, there exists $(x^\varepsilon, p^\varepsilon)$ and $\alpha^\varepsilon$ such that
\begin{enumerate}
    \item $x^\varepsilon \in \mathcal{M}$
    \item $x_i^\varepsilon \in \arg\max_{y_i \in \Delta_-^n} \{u_i \cdot y_i : p^\varepsilon \cdot y_i \leq p^\varepsilon \cdot \omega^\varepsilon_i + \alpha^\varepsilon\}$ for each agent $i$
    \item $x_i \in \arg\min_{y_i \in \Delta_-^n} \{p^\varepsilon \cdot y_i : u_i \cdot y_i \geq u_i \cdot x_i^\varepsilon\}$ for each agent $i$
\end{enumerate}
Now, let $p^{\varepsilon, 1} \coloneqq p^\varepsilon$.
For $k \in [n]$, iteratively define $r^{\varepsilon, k}_{j\ell} \coloneqq p^{\varepsilon, k}_j / |p^{\varepsilon, k}_\ell|$.
Let $r^{(k)}_{j\ell}$ denote the limit of $r^{\varepsilon,k}_{j\ell}$ (after passing to a subsequence).
Consider the total preorder $\succeq_k$ on $[n]$ given by $j \succeq_k \ell$ if and only if $r^{(k)}_{j\ell} \not= 0$.
\begin{enumerate}
    \item (Reflexivity) $j \succeq_k j$ since $r^{(k)}_{jj} = \lim_{\varepsilon \to 0} p^{\varepsilon, k}_j / p^{\varepsilon, k}_j = 1 \not= 0$.
    \item (Transitivity) If $j \succeq_k \ell'$ (so $r^{(k)}_{j\ell'} \not= 0$) and $\ell' \succeq_1 \ell$ (so $r^{(k)}_{\ell'\ell} \not= 0$), then 
    \[
        r^{(k)}_{j\ell} 
            = \lim_{\varepsilon \to 0} p^{\varepsilon,k}_j / p^{\varepsilon,k}_\ell 
            = \lim_{\varepsilon \to 0} p^{\varepsilon,1}_j / p^{\varepsilon,k}_{\ell'} \cdot \lim_{\varepsilon \to 0} p^{\varepsilon,k}_{\ell'} / p^{\varepsilon,k}_\ell 
            = r^{(1)}_{j\ell'} \cdot r^{(1)}_{\ell'\ell} 
            \not= 0
    \]
    \item (Totality) Either $r^{(k)}_{j\ell} \not= 0$ (so $j \succeq_k \ell$) or $r^{(k)}_{j \ell} = 0$.
    We are in the latter case only if $p^{\varepsilon,k}_j = 0$ while $p^{\varepsilon,k}_\ell \not= 0$, so $r^{(k)}_{\ell j} \not= 0$ and $j \preceq_k \ell$.
\end{enumerate}
Consequently, the binary relation $\sim_k$ given by $j \sim_k \ell$ if and only if $j \succeq_k \ell$ and $j \preceq_k \ell$ constitutes an equivalence relation, and $\succeq_k$ induces a total order on the quotient set $[n] / \sim_k$.
Let $j_k$ denote the representative of the greatest element of $[n] / \sim_k$ under $\succeq_k$, and define $C^\varepsilon_k \coloneqq |p^{\varepsilon,k}_{j_k}|$. 
For all items $j$, define $p^{(k)}_j \coloneqq r^{(k)}_{jj_k}$ and $p^{\varepsilon, k+1}_j \coloneqq p^{\varepsilon, k}_j - C^\varepsilon_{k} \cdot p^{(k)}_j$.

\begin{lemma}\label{lemma:some-properties}
Some properties:
\begin{enumerate}
    \item If $p^{\varepsilon, k}_j = 0$, then $p^{\varepsilon, k+1}_j = 0$.
    \item $C^\varepsilon_k \cdot p^{(k)}_{j_k} = p^{\varepsilon, k}_{j_k}$ for all $k \in [n]$.
    \item $p^{\varepsilon, k}_{j_\ell} = 0$ for all $\ell < k \in [n]$.
    \item $p^{(k)}_{j_\ell} = 0$ for all $\ell < k \in [n]$ such that $C^\varepsilon_k > 0$ for all sufficiently small $\varepsilon$.
    \item For all items $j$, $p^{\varepsilon, n+1}_j = 0$.
    \item $p^{(k)}_j \not= 0$ if and only if $j \sim_k j_k$.
    \item $p^{(k)}_j < 0$ only if $j \sim_\ell j_\ell$ for some $\ell \in [k-1]$.
    \item $p^{(k)}_j \in \RR$ for all items $j$ and $k \in [n]$.
    \item $\lim_{\varepsilon \to 0} C^\varepsilon_{k+1} / C^\varepsilon_{k} = 0$ for all $k \in [n]$ such that $C^\varepsilon_k > 0$ for all sufficiently small $\varepsilon$. 
    \item If $p^{(\ell)} \cdot y_i \leq 0$ for all $\ell \leq k$, then $y_{ij} = 0$ for all $j \sim_\ell j_\ell$ for some $\ell \leq k$.
\end{enumerate}
\end{lemma}

\begin{proof}
(1) follows from the fact that
\[
    p^{\varepsilon, k+1}_j = p^{\varepsilon, k}_j - C^\varepsilon_k \cdot p^{(k)}_j = 0 - C^\varepsilon_k \cdot \lim_{\varepsilon \to 0} 0/C^\varepsilon_k = 0
\]
while (2) follows from the fact that
\[
    C^\varepsilon_k \cdot p^{(k)}_{j_k} = |p^{\varepsilon, k}_{j_k}| \cdot r^{(k)}_{j_kj_k} = |p^{\varepsilon, k}_{j_k}| \cdot \mathrm{sgn}(p^{\varepsilon, k}_{j_k}) = p^{\varepsilon, k}_{j_k}
\]
(3) follows from (1) and (2): by (2), $p^{\varepsilon, \ell+1}_{j_\ell} = p^{\varepsilon,\ell+1}_{j_\ell} - C^\varepsilon_\ell \cdot p^{(\ell)}_{j_\ell} = 0$, which, by repeated application of (1), implies that $p^{\varepsilon, k}_{j_\ell} = 0$ for all $k > \ell$.
(4) follows from (3): if $C^\varepsilon_k > 0$ for all sufficiently small $\varepsilon$, then 
\[
    r^{\varepsilon,k}_{j_\ell j_k} = p^{\varepsilon, k}_{j_\ell} / C^\varepsilon_k = 0
\]
and $p^{(k)}_{j_\ell} = r^{(k)}_{j_\ell j_k} = 0$ for all $\ell \in [k-1]$.

To see (5), suppose by way of contradiction that there exists an item $j$ such that $p^{\varepsilon, n+1}_j \not= 0$.
By (2) and (3), it must be the case that $j \not\in \{j_1, \dots, j_n\}$.
Since there are only $n$ items, there exists $\ell < k \in [n]$ such that $j_\ell = j_k$.
By (3), we know that $p^{\varepsilon, k}_{j_\ell} = 0$.
But since $j_\ell$ represents the greatest element of $[n] / \sim_k$ under $\succeq_k$, we know that $r^{(k)}_{j_\ell j} \not= 0$, so it must be the case that $p^{\varepsilon, k}_j = 0$.
By (1), it follows that $p^{\varepsilon, n+1}_j = 0$, a contradiction.

(6) follows from the definition of $\succeq_k$.
To see (7), note that by (6), $j \prec_\ell j_\ell$ for all $\ell \in [k-1]$ implies that $p^{(\ell)}_j = 0$ for all such $\ell$.
Thus, $p^{\varepsilon, k}_j = p^\varepsilon_j \geq 0$, so $p^{(k)}_j = \lim_{\varepsilon \to 0} p^{\varepsilon, k}_j / C^\varepsilon_k \geq 0$.
(8) follows from the definition of $\succeq_k$ and $j_k$.

To see (9), note that for all $\delta > 0$, we have that
\[  
    p^{(k)}_{j_{k+1}} = p^{\varepsilon, k}_{j_{k+1}} / C^\varepsilon_k \pm \delta
\]
for sufficiently small $\varepsilon$ since $p^{(k)}_{j_{k+1}} \in \RR$ by definition of $\succeq_k$ and $j_k$.
Thus, we will eventually have that
\[
    \frac{C^\varepsilon_{k+1}}{C^\varepsilon_k} = \frac{|p^{\varepsilon, k}_{j_{k+1}} - C^\varepsilon_k \cdot p^{(k)}_{j_{k+1}}|}{C^\varepsilon_k} \leq \delta
\]

To see (10), suppose by way of contradiction that there exists an item $j \sim_\ell j_\ell$ for some $\ell \leq k$ yet $y_{ij} > 0$.
Let $\ell$ denote the smallest index for which such an item exists, and let $j$ denote the item.
By the minimality of $\ell$, it follows from (6) and (7) that 
\begin{enumerate}
    \item $p^{(\ell)}_j > 0$ (otherwise, by (6), $p^{(\ell)}_j < 0$, so by (7), $j \sim_h j_h$ for some $h < \ell$, contradicting the minimality of $\ell$) 
    \item $p^{(h)} \cdot y_i = 0$ for all $h < \ell$ (otherwise, there exists $j' \in \supp{y_i}$ such that $p^{(h)}_{j'} \not= 0$ for some $h < \ell$, which, by (6), implies that $j' \sim_h j_h$, contradicting the minimality of $\ell$)
    \item $p^{(k)}_{j'} \geq 0$ for all $j' \in \supp{y_i}$ (otherwise, by (7), $j' \sim_h j_h$ for some $h < \ell$, contradicting the minimality of $\ell$)
\end{enumerate}
It follows that $p^{(\ell)} \cdot y_i > 0$, contradicting the fact that the opposite is true.

\end{proof}

Now, define
\[
    M \coloneqq \sup \{k \in [n] : \exists\,\delta : C^\varepsilon_k > 0 \,\forall\, \varepsilon < \delta\}
\]
If $M = -\infty$, then $p^\varepsilon = 0$ infinitely often, so it is possible to clear the market while giving each agent her favorite lottery for free.
For the rest of the proof, we assume that $M \in [n]$.

\begin{lemma}\label{lemma:decomposition-of-prices}
For all items $j$, $p^\varepsilon_j = \sum_{k=1}^M C^\varepsilon_k \cdot p^{(k)}_j$.
\end{lemma}

\begin{proof}
Note that
\[
    \textstyle p^\varepsilon_j = p^{\varepsilon,1}_j = p^{\varepsilon, 2}_j + C^\varepsilon_1 \cdot p^{(1)}_j = \dots = p^{\varepsilon, M+1}_j + \sum_{k=1}^M C^\varepsilon_k \cdot p^{(k)}_j  
\]
If $M = n$, then by (5) in Lemma~\ref{lemma:some-properties}, $p^{\varepsilon, M+1}_j = 0$.
Otherwise, $M < n$, and by definition of $M$, $|p^{\varepsilon, M+1}_{j_{M+1}}| = C^\varepsilon_{M+1} = 0$.
Since $j_{M+1}$ represents the greatest element of $[n] / \sim_{M+1}$ under $\succeq_{M+1}$, we know that $\lim_{\varepsilon \to 0} p^{\varepsilon, M+1}_{j_{M+1}} / |p^{\varepsilon, M+1}_{j}| = r^{(M+1)}_{j_{M+1} j} \not= 0$ for all items $j$, so it must be the case that $p^{\varepsilon, M+1}_j = 0$.
\end{proof}

Now, define
\begin{align*}
    m &\coloneqq M \wedge \inf \CrBr{k \in [M]: \lim_{\varepsilon \to 0} \frac{\frac{\varepsilon}{n} \sum_j p^\varepsilon_j + \alpha^\varepsilon}{C^\varepsilon_k} > 0} \\
    \alpha &\coloneqq \lim_{\varepsilon \to 0} \frac{\frac{\varepsilon}{n} \sum_j p^\varepsilon_j + \alpha^\varepsilon}{C^\varepsilon_m}
\end{align*}
We show that $(x, p^{(1)}, \dots, p^{(m)})$ where $x$ denotes the limit of $x^\varepsilon$ (after passing to a subsequence) constitutes an $\alpha$-EDE in $m$ lexicographic currencies.

For each agent $i$, define 
\[
    k_i \coloneqq M \wedge \inf \CrBr{k \in [M]: \lim_{\varepsilon \to 0} \frac{p^\varepsilon \cdot \omega^\varepsilon_i + \alpha^\varepsilon}{C^\varepsilon_k} > 0}
\]
Note that $k_i \leq m$ for all agents $i$ since $\omega^\varepsilon_i \coloneqq \frac{\varepsilon}{m} \cdot \1 + (1 - \varepsilon) \cdot \omega_i$.

\begin{lemma}\label{lemma:no-higher-currencies}
For all $k < k_i$, $p^{(k)} \cdot \omega_i = 0$.
\end{lemma}

\begin{proof}
Suppose by way of contradiction otherwise.
Let $k$ denote the smallest index for which $p^{(k)} \cdot \omega_i \not= 0$.
By (10) in Lemma~\ref{lemma:some-properties}, $\omega_{ij} = 0$ for all $j \sim_\ell j_\ell$ for some $\ell < k$.
It follows that $p^{(k)} \cdot \omega_i > 0$ as otherwise, there would exist $j \in \supp{\omega_i}$ such that $j \sim_\ell j_\ell$ for some $\ell < k$ by (7) in Lemma~\ref{lemma:some-properties}.
It follows that
\begin{align*}
    \lim_{\varepsilon \to 0} \frac{p^\varepsilon \cdot \omega^\varepsilon_i + \alpha^\varepsilon}{C^\varepsilon_k} 
        &\geq \lim_{\varepsilon \to 0} \frac{(1-\varepsilon) \cdot p^\varepsilon \cdot \omega_i}{C^\varepsilon_k} \tag{definition of $\omega^\varepsilon_i$} \\
        &= \lim_{\varepsilon \to 0} \frac{(1-\varepsilon) \sum_{\ell=1}^M C^\varepsilon_\ell \cdot p^{(\ell)} \cdot \omega_i}{C^\varepsilon_k} \tag{Lemma~\ref{lemma:decomposition-of-prices}} \\
        &= \lim_{\varepsilon \to 0} \frac{(1-\varepsilon) \sum_{\ell=k}^M C^\varepsilon_\ell \cdot p^{(\ell)} \cdot \omega_i}{C^\varepsilon_k} \tag{minimality of $k$} \\
        &\geq p^{(k)} \cdot \omega_i > 0 \tag{(8) and (9) in Lemma~\ref{lemma:some-properties}; definition of $M$}
\end{align*}
which contradicts the minimality of $k_i$.
\end{proof}

\begin{lemma}\label{lemma:endowment-in-each-currency-equals-limit-of-endowment-v2}
For all $k \leq k_i$, 
\[
    p^{(k)} \cdot \omega_i + \alpha \cdot \1(k = m) = \lim_{\varepsilon \to 0} \frac{p^\varepsilon \cdot \omega^\varepsilon_i + \alpha^\varepsilon}{C^\varepsilon_k}
\]
\end{lemma}

\begin{proof}
For all $k \leq k_i$,
\begin{align*}
    \lim_{\varepsilon \to 0} \frac{p^\varepsilon \cdot \omega^\varepsilon_i + \alpha^\varepsilon}{C^\varepsilon_k} 
        &= \lim_{\varepsilon \to 0} \frac{(1-\varepsilon) \sum_{\ell=1}^M C^\varepsilon_\ell \cdot p^{(\ell)} \cdot \omega_i + \frac{\varepsilon}{n} \sum_j p^\varepsilon_j + \alpha^\varepsilon}{C^\varepsilon_k} \tag{definition of $\omega_i$; Lemma~\ref{lemma:decomposition-of-prices}} \\
        &= \lim_{\varepsilon \to 0} \frac{(1-\varepsilon) \sum_{\ell=k_i}^M C^\varepsilon_\ell \cdot p^{(\ell)} \cdot \omega_i + \frac{\varepsilon}{n} \sum_j p^\varepsilon_j + \alpha^\varepsilon}{C^\varepsilon_k} \tag{Lemma~\ref{lemma:no-higher-currencies}} \\
        &= (p^{(k_i)} \cdot \omega_i + \alpha \cdot \1(k = m)) \cdot \1(k = k_i) \tag{(8) and (9) in Lemma~\ref{lemma:some-properties}; definition of $M$; definition of $\alpha$}
\end{align*}
To conclude, note that by Lemma~\ref{lemma:no-higher-currencies}, $p^{(k)} \cdot \omega_i = 0$ for all $k < k_i$.
\end{proof}

\begin{lemma}\label{lemma:limiting-allocation-feasible}
$x_i \in C^\alpha_i(p)$.
\end{lemma}

\begin{proof}
Suppose by way of contradiction that there exists $k \leq k_i$ such that $p^{(k)} \cdot x_i > p^{(k)} \cdot \omega_i + \alpha \cdot \1(k = m)$.
Let $k$ denote the smallest such index, and let $\gamma > 0$ denote the gap.
Note that the minimality of $k$ implies that $p^{(\ell)} \cdot x_i \leq p^{(\ell)} \cdot \omega_i = 0$ (where the last equality follows from Lemma~\ref{lemma:no-higher-currencies}).
By (10) in Lemma~\ref{lemma:some-properties}, $x_{ij} = 0$ for all $j \sim_\ell j_\ell$ for some $\ell < k$.
It follows that
\begin{align*}
    p^{(k)} \cdot \omega_i + \alpha \cdot \1(k = m)
        &= \lim_{\varepsilon \to 0} \frac{p^\varepsilon \cdot \omega_i^\varepsilon + \alpha^\varepsilon}{C^\varepsilon_k} \tag{Lemma~\ref{lemma:endowment-in-each-currency-equals-limit-of-endowment-v2}} \\ 
        &\geq \lim_{\varepsilon \to 0} \frac{p^\varepsilon \cdot x_i^\varepsilon}{C^\varepsilon_k} \tag{feasibility of $x^\varepsilon_i$} \\
        &= \lim_{\varepsilon \to 0} \frac{\sum_{j : \exists\, \ell < k : j \sim_\ell j_\ell} p^{\varepsilon}_j \cdot x^\varepsilon_{ij} + \sum_{h=1}^M C^\varepsilon_h \cdot \sum_{j : \forall\, \ell < k : j \prec_\ell j_\ell} p^{(h)}_j \cdot x^\varepsilon_{ij}}{C^\varepsilon_k} \tag{Lemma~\ref{lemma:decomposition-of-prices}} \\
        &\geq \lim_{\varepsilon \to 0} \frac{\sum_{h=k}^M C^\varepsilon_h \cdot \sum_{j : \forall\, \ell < k : j \prec_\ell j_\ell} p^{(h)}_j \cdot x^\varepsilon_{ij}}{C^\varepsilon_k} \tag{(6) in Lemma~\ref{lemma:some-properties}} \\
        &= \lim_{\varepsilon \to 0} \textstyle \sum_{j : \forall\, \ell < k : j \prec_\ell j_\ell} p^{(k)}_j \cdot x^\varepsilon_{ij} \tag{(8) and (9) in Lemma~\ref{lemma:some-properties}; definition of $M$} \\
        &\geq \textstyle \sum_{j : \forall\, \ell < k : j \prec_\ell j_\ell} p^{(k)}_j \cdot x_{ij} \\
        &= p^{(k)} \cdot x_i \tag{$x_{ij} = 0$ for all $j \sim_\ell j_\ell$ for some $\ell < k$}
\end{align*}
Note that this conclusion contradicts the definition of $k$.
\end{proof}

\begin{lemma}\label{lemma:limiting-allocation-optimal}
$x_i \in \arg\max \{u_i \cdot y_i : y_i \in C^\alpha_i(p)\}$.
\end{lemma}

\begin{proof}
Membership follows from Lemma~\ref{lemma:limiting-allocation-feasible}.
To prove optimality, suppose by way of contradiction that there exists $x'_i \in C^\alpha_i(p)$ such that $u_i \cdot x'_i > u_i \cdot x_i$.
Note that membership in $C^\alpha_i(p)$ implies that $p^{(k)} \cdot x'_i \leq p^{(k)} \cdot \omega_i = 0$ for all $k < k_i$ where the equality follows from Lemma~\ref{lemma:no-higher-currencies}.
Consequently, by (10) in Lemma~\ref{lemma:some-properties}, we have that $x'_{ij} = 0$ for all $j \sim_k j_k$ for some $k < k_i$, so $p^{(k)} \cdot x'_i = 0$ for all $k < k_i$.
Now, define $\gamma > 0$ so that $(1 - \gamma) \cdot u_i \cdot x'_i > u_i \cdot x_i$.
We show that for sufficiently small $(1 - \gamma) \cdot p^\varepsilon \cdot x_i' \leq p^\varepsilon \cdot \omega_i^\varepsilon + \alpha^\varepsilon$.
Note that this inequality suffices for a contradiction to the optimality of $x^\varepsilon_i$ since for sufficiently small $\varepsilon$, $(1 - \gamma) \cdot u_i \cdot x_i' > u_i \cdot x^\varepsilon_i$.

We break the proof into three cases.
In the first case, either $p^{(k_i)} \cdot \omega_i > 0$ or $\alpha \in (0, +\infty)$.
Note that these conditions imply that $p^{(k_i)} \cdot \omega_i + \alpha \cdot \1(k_i = m) > 0$: if $p^{(k_i)} \cdot \omega_i = 0$, then $k_i = m$ by Lemma~\ref{lemma:endowment-in-each-currency-equals-limit-of-endowment-v2} and the definition of $k_i$, so $p^{(k_i)} \cdot \omega_i + \alpha \cdot \1(k_i = m) = \alpha > 0$.
Now, since $p^{(k_i)} \cdot \omega_i + \alpha \cdot \1(k_i = m) > 0$, (9) in Lemma~\ref{lemma:some-properties} and the definition of $M$ imply that
\begin{equation}\label{eq:exchange-rates}
    C^\varepsilon_k / C^\varepsilon_{k_i} \leq \frac{\gamma \cdot (p^{(k_i)} \cdot \omega_i + \alpha \cdot \1(k_i = m))}{3M \cdot (1 + \max_k | p^{(k)} \cdot x'_i | + \max_k | p^{(k)} \cdot \omega_i |)}
\end{equation}
for all $k \in \{k_i + 1, \cdots, M\}$ and sufficiently small $\varepsilon$.
For all such $\varepsilon$, we also have that
\begin{equation}\label{eq:surplus}
    \abs{\alpha - \frac{\frac{\varepsilon}{n} \sum_j p^\varepsilon_j + \alpha^\varepsilon}{C^\varepsilon_m}} \leq \gamma \cdot \alpha / 3
\end{equation}
by the definition of $\alpha$.
Thus,
\begin{align*}
    (1 - \gamma) \cdot p^\varepsilon \cdot x'_i
        &= \textstyle (1 - \gamma) \sum_{k=1}^M C^\varepsilon_k \cdot p^{(k)} \cdot x'_i \tag{Lemma~\ref{lemma:decomposition-of-prices}} \\
        &= \textstyle (1 - \gamma) \sum_{k=k_i}^M C^\varepsilon_k \cdot p^{(k)} \cdot x'_i \tag{$p^{(k)} \cdot x'_i = 0$ for all $k < k_i$} \\
        &\leq \textstyle (1 - \gamma) \cdot C^\varepsilon_{k_i} \cdot (p^{(k_i)} \cdot \omega_i + \alpha \cdot \1(k_i = m)) + \sum_{k=k_i+1}^M C^\varepsilon_k \cdot | p^{(k)} \cdot x'_i | \tag{$x'_i \in C^\alpha_i(p)$} \\
        &\leq \textstyle (1 - \gamma / 3) \cdot \sum_{k=k_i}^M C^\varepsilon_{k} \cdot p^{(k)} \cdot \omega_i + (1 - \gamma / 3) \cdot C^\varepsilon_{k_i} \cdot \alpha \cdot \1(k_i = m) \tag{Equation~\ref{eq:exchange-rates}} \\
        &= \textstyle (1 - \gamma / 3) \cdot \sum_{k=1}^M C^\varepsilon_{k} \cdot p^{(k)} \cdot \omega_i + (1 - \gamma / 3) \cdot C^\varepsilon_{k_i} \cdot \alpha \cdot \1(k_i = m) \tag{Lemma~\ref{lemma:no-higher-currencies}} \\
        &= \textstyle (1 - \gamma / 3) \cdot p^\varepsilon \cdot \omega_i + (1 - \gamma / 3) \cdot C^\varepsilon_{k_i} \cdot \alpha \cdot \1(k_i = m) \tag{Lemma~\ref{lemma:decomposition-of-prices}} \\
        &\leq \textstyle (1 - \gamma / 3) \cdot p^\varepsilon \cdot \omega_i + \frac{\varepsilon}{n} \sum_j p^\varepsilon_j + \alpha^\varepsilon \tag{Equation~\ref{eq:surplus}} \\
        &\leq p^\varepsilon \cdot \omega^\varepsilon_i + \alpha^\varepsilon \tag{for sufficiently small $\varepsilon$}
\end{align*}

In the second case, $p^{(k_i)} \cdot \omega_i = 0$ and $\alpha = +\infty$.
It follows from Lemma~\ref{lemma:endowment-in-each-currency-equals-limit-of-endowment-v2} and the definition of $k_i$ that $k_i = m$. 
Thus, $\alpha = +\infty$ implies that
\[
    \textstyle \frac{\varepsilon}{n} \sum_j p^\varepsilon_j + \alpha^\varepsilon > C^\varepsilon_{k_i} \cdot \sum_{k=k_i}^M |p^{(k)} \cdot x'_i|
\]
for all sufficiently small $\varepsilon$.
By definition of $M$, $C^\varepsilon_{k_i} > 0$ for all such $\varepsilon$ as well, so
\begin{align*}
    p^\varepsilon \cdot x'_i
        &= \textstyle \sum_{k=1}^M C^\varepsilon_k \cdot p^{(k)} \cdot x'_i \tag{Lemma~\ref{lemma:decomposition-of-prices}} \\
        &= \textstyle \sum_{k=k_i}^M C^\varepsilon_k \cdot p^{(k)} \cdot x'_i \tag{$p^{(k)} \cdot x'_i = 0$ for all $k < k_i$} \\ 
        &\leq \textstyle C^\varepsilon_{k_i} \cdot \sum_{k=k_i}^M |p^{(k)} \cdot x'_i| \tag{(9) in Lemma~\ref{lemma:some-properties}; definition of $M$} \\
        &< p^\varepsilon \cdot \omega_i^\varepsilon + \alpha^\varepsilon
\end{align*}

In the last case, $p^{(k_i)} \cdot \omega_i = \alpha = 0$.
It follows that $k_i = m = M$ where the first equality follows from Lemma~\ref{lemma:endowment-in-each-currency-equals-limit-of-endowment-v2} and the definition of $k_i$ while the latter equality follows from the definition of $m$.
Since $x'_i \in C^\alpha_i(p)$, we know that $p^{(k)} \cdot x_i \leq p^{(k)} \cdot \omega_i + \alpha \cdot \1(k = m) = 0$.
Thus,
\[
    p^\varepsilon \cdot x'_i
        = \textstyle \sum_{k=1}^M C^\varepsilon_k \cdot p^{(k)} \cdot x'_i 
        \leq 0 
        \leq p^\varepsilon \cdot \omega_i^\varepsilon + \alpha^\varepsilon
\]
\end{proof}

\begin{lemma}\label{lemma:strong-cheapest-bundle}
There does not exist $x_i' \in C^\alpha_i(p)$ such that $u_i \cdot x_i' \geq u_i \cdot x_i$ yet there exists $k$ such that $p^{(\ell)} \cdot x_i' \leq p^{(\ell)} \cdot x_i$ for all $\ell < k$ and $p^{(k)} \cdot x_i' < p^{(k)} \cdot x_i$.
\end{lemma}

\begin{proof}
Suppose by way of contradiction that such an $x'_i$ exists.
Let $\gamma = p^{(k)} \cdot (x_i - x'_i)$.
(9) in Lemma~\ref{lemma:some-properties} and the definition of $M$ imply that
\[
    C^\varepsilon_\ell / C^\varepsilon_{k} \leq \frac{\gamma}{3M \cdot (1 + \max_\ell | p^{(\ell)} \cdot x'_i | + \max_\ell | p^{(\ell)} \cdot x_i |)}
\]
for all $\ell > k$ and sufficiently small $\varepsilon$.
It follows that
\begin{align*}
    p^\varepsilon \cdot x'_i 
        &= \textstyle \sum_{\ell=1}^M C^\varepsilon_\ell \cdot p^{(\ell)} \cdot x'_i \tag{Lemma~\ref{lemma:decomposition-of-prices}} \\
        &\leq \textstyle \sum_{\ell=1}^k C^\varepsilon_\ell \cdot p^{(\ell)} \cdot x_i - C^\varepsilon_k \cdot \gamma + \sum_{\ell=k+1}^M C^\varepsilon_\ell \cdot p^{(\ell)} \cdot x'_i \tag{definition of $\gamma$} \\
        &\leq \textstyle \sum_{\ell=1}^M C^\varepsilon_\ell \cdot p^{(\ell)} \cdot x_i - C^\varepsilon_k \cdot \gamma / 3 \\
        &< p^\varepsilon \cdot x_i
\end{align*}
for all sufficiently small $\varepsilon$.
Now, note that for all such $\varepsilon$, we also have that $x_i / 2 \leq x_i^\varepsilon$.
This is because eventually, $\lvert x_{ij} - x^\varepsilon_{ij}\rvert < \min_{j : x_{ij} > 0} x_{ij} / 2$ for all items $j$.
But consider $y_i \coloneqq x^\varepsilon_i - (x_i - x'_i)/2$ (which is a feasible allocation since $\sum_j x_{ij} = 1$ while $\sum_j x'_{ij} \leq 1$): since $u_i \cdot x_i \leq u_i \cdot x'_i$, we have that $u_i \cdot y_i \geq u_i \cdot x_i^\varepsilon$, yet since $p^\varepsilon \cdot x_i > p^\varepsilon \cdot x'_i$, we also have that $p^\varepsilon \cdot y_i < p^\varepsilon \cdot x_i^\varepsilon$, contradicting the fact that $x^\varepsilon_i$ is the cheapest lottery with utility at least $u_i \cdot x^\varepsilon_i$.
\end{proof}

\section{Core Convergence: Proof of Theorem~\ref{theorem:weak-LDE-contains-rejective-core}}

In this section, we formally verify that the set of LDEs with simple prices satisfying the weak and aggregate cheapest bundle properties contains the limit of the rejective core as the economy grows.

\weakLDEcontainsRC*

\begin{proof}
Let $x \in \bigcap_{N=1}^\infty RC(u, \omega, N)$, and consider the following procedure for constructing $p$ and $\alpha$ recursively.
\begin{enumerate}
    \item $S \gets [m]$, $T \gets \varnothing$, $k \gets 0$.
    \item If $x$ satiates each agent, then $p \gets 0$, $\alpha \gets 0$, and $S \gets \varnothing$.
    \item While $S \not= \varnothing$,
    \begin{enumerate}
        \item $k \gets k+1$.
        \item For each agent $i \in T$, let $Q_i(S) \coloneqq \{z_i \in \RR^n : z_i + x_i \succeq_i x_i\}$
        \item Let $Q_T(S)$ denote the intersection of $\{z \in \RR^n : z_{j} \leq 0 \,\forall\, j \not\in S\}$ and the convex hull of $\cup_{i \in T} Q_i(S)$.
        \item For each agent $i \not\in T$, let $P_i(S) \coloneqq \{z_i \in \RR^n : z_i + \omega_i \succ_i x_i \wedge z_{ij} = 0 \,\forall\, j \not\in S\}$ and $Q_i(S) \coloneqq \{z_i \in \RR^n : z_i + x_i \succeq_i x_i \wedge z_{ij} = 0 \,\forall\, j \not\in S\}$.
        \item Let $P(S)$ denote the convex hull of $\cup_{i \not\in T} (P_i(S) \cup Q_i(S)) \cup Q_T(S)$.
        \item Find $p^{(k)} \in \RR^{n} \setminus \{0\}$ so that $p^{(k)} \cdot z \geq 0$ for all $z \in P(S)$ and $p^{(k)} \cdot z \leq 0$ for all $z \in V(S)$ where $V(S) \coloneqq \{z \in \RR^n : z_{j} < 0 \,\forall\, j \in S\}$.
        That is, $\{z \in \RR^n : p^{(k)} \cdot z = 0\}$ is a hyperplane that separates $P(S)$ from $V(S)$. We will prove that such a hyperplane exists.
        \item For each agent $i$, $\alpha^{(k)}_i \gets \max\{p^{(k)} \cdot (x_i - \omega_i), 0\}$.
        \item $S \gets S \setminus \{j \in [n]: p^{(k)}_j > 0\}$.
        \item $T \gets T \cup \{i \in [n] : p^{(k)} \cdot \omega_i + \alpha^{(k)}_i > 0\}$.
        \item If $y_i \preceq_i x_i$ for all $i \not\in T$ and $y_i$ such that $p^{(\ell)} \cdot y_i = 0$ for all $\ell \in [k]$, that is, if agents with no income by the end of this iteration receive their favorite free bundle, then $S \gets \varnothing$.
    \end{enumerate}
\end{enumerate}
We prove by induction on $k$ that $(x, p, \alpha)$ constitutes an LDE with simple prices that satisfies the weak and the aggregate cheapest bundle properties.
Let $S_k$ and $T_k$ denote $S$ and $T$, respectively, at the beginning of the $k$-th iteration.
By \ref{item:IH-non-negative-prices},
\begin{align*}
    S_k 
        &= \{j \in [n] : p^{(\ell)}_j = 0 \,\forall\, \ell \in [k-1]\} \\
    T_k 
        &= \{i \in [n] :  \exists\,\ell \in [k-1] : p^{(\ell)} \cdot \omega_i + \alpha^{(\ell)}_i > 0\}
\end{align*}
In words, $S_k$ is the set of free items at the beginning of the $k$-th iteration, while $T_k$ is the set of agents with positive income by the beginning of the $k$-th iteration.
Our inductive hypothesis is the following:
\begin{enumerate}[label=\text{IH.\arabic*}]
    \item\label{item:IH-non-negative-prices} $p^{(\ell)} \in \RR^n_+$ for all $\ell \in [k]$.
    \item\label{item:IH-simple-prices} For all items $j$, $p^{(\ell)}_j \not= 0$ for some $\ell \in [k]$ implies $p^{(h)}_j = 0$ for all $h \in [k] \setminus \{\ell\}$.
    \item\label{item:IH-optimal-consumption} For all $i \in T_{k+1}$, there does not exist $y_i \in C^\alpha_i(p)$ such that $y_i \succ_i x_i$. 
    Note that $C^\alpha_i(p)$ is well-defined as soon as an agent has positive income.
    \item\label{item:IH-weak-cheapest-bundle} For all $i \in T_{k+1}$, there does not exist $y_i \succeq_i x_i$, and $k \in [k_i]$ such that $p^{(\ell)} \cdot y_i = p^{(\ell)} \cdot x_i$ for all $\ell \in [k-1]$ and $p^{(k)} \cdot y_i < p^{(k)} \cdot x_i$.
    Note that $k_i$ is well-defined as soon as an agent has positive income.
    \item\label{item:IH-agg-cheapest-bundle} There does not exist $\beta \in \RR^n_+$, $y \in X$, and $\ell \in [k]$ such that 
    \begin{itemize}
        \item $y_i \succeq_i x_i$ for all $i \in \supp{\beta}$
        \item $p^{(h)} \cdot \sum_i \beta_i y_i = p^{(h)} \cdot \sum_i \beta_i x_i$ for all $h \in [\ell-1]$ and $p^{(\ell)} \cdot \sum_i \beta_i  y_i < p^{(\ell)} \cdot \sum_i \beta_i x_i$
        \item $\sum_i \beta_i y_{ij} \leq \sum_i \beta_i x_{ij}$ for all $j \not\in S_{\ell}$.
    \end{itemize}
    \item\label{item:IH-progress} If $S_k \not= \varnothing$, then there exists $j \in S_k$ such that $p^{(k)}_j > 0$.
\end{enumerate}
Note that \ref{item:IH-progress} implies that our procedure terminates in at most $n$ iterations.

It is clear that our inductive hypothesis holds for $k = 0$ (define $S_0 \coloneqq \varnothing$).
We now show that if it holds for $k$, then it holds for $k+1$.
Note that if $S_{k+1} = \varnothing$, then our inductive hypothesis implies that $(x, p, \alpha)$ constitutes an LDE: any agent $i \in T_{k+1}$ weakly prefers $x_i$ over her other affordable bundles by \ref{item:IH-optimal-consumption}, while any agent $i \not\in T_{k+1}$ must be receiving her favorite free bundle (otherwise, the set of free items after the $k$-th iteration $S_{k+1} \not= \varnothing$).
Meanwhile, \ref{item:IH-non-negative-prices} implies that the first non-zero entry in $p_j$ is positive for all $j$, and \ref{item:IH-simple-prices} states that the prices are simple.
\ref{item:IH-weak-cheapest-bundle} and \ref{item:IH-agg-cheapest-bundle} guarantee the weak and aggregate cheapest bundle properties, respectively.
We carry out the inductive step under the assumption that $S_{k+1} \not= \varnothing$. 

We begin by showing that $P(S_{k+1}) \cap V(S_{k+1}) = \varnothing$.
Let $U$ denote the set of agents for which there exists $y_i \succ_i x_i$ such that 
\[
    \begin{cases}
        y_i \in C^\alpha_i(p) & i \in T_{k+1} \\
        p^{(\ell)} \cdot y_i = 0 \,\forall\, \ell \in [k] & i \not\in T_{k+1}
    \end{cases}
\]
Note that by definition, $S_{k+1} \not= \varnothing$ implies that $U \not= \varnothing$.
Moreover, by \ref{item:IH-optimal-consumption}, $U$ is precisely the set of agents $i \not\in T_{k+1}$ for which there exists $y_i \succ_i x_i$ such that $p^{(\ell)} \cdot y_i = 0$ for all $\ell \in [k]$.
It follows that $P_i(S_{k+1}) \not= \varnothing$ iff $i \in U$.
Now, suppose $P(S_{k+1}) \cap V(S_{k+1}) \not= \varnothing$.
In particular, since $\succeq_i$ is convex for all agents $i$, there exists $(\lambda, \mu) \in \Delta^{\abs{U} + n}$, $z_i \in P_i(S_{k+1})$ for all $i \in U$, and $z'_i \in Q_i(S_{k+1})$ for all $i \in [n]$ such that 
\begin{align*}
    & z_{ij} = 0 \,\forall\, i \in U, j \not\in S_{k+1} \\
    & z'_{ij} = 0 \,\forall\, i \not\in T_{k+1}, j \not\in S_{k+1} \\
    & \sum_{i \in T} \mu_i z'_{ij} \leq 0 \,\forall\, j \not\in S_{k+1} \\
    & \sum_{i \in U} \lambda_i z_i + \sum_i \mu_i z'_i \in V(S_{k+1}) 
\end{align*}
by definition.
We assume without loss of generality that $(\lambda, \mu) \in \Delta^{\abs{U} + n} \cap \QQ^{\abs{U} + n}$ and that $\lambda_i > 0$ for some $i \in U$ since a tiny shift in the convex coefficients would not cause the convex combination to leave $V(S_{k+1})$ (which is an open set).
Let $N \in \NN$ be such that $N \cdot (\lambda, \mu) \in (\NN \cup \{0\})^{\abs{U} + n}$, and consider the coalition $C$ with $\lambda_i N$ copies of agent $i \in U$ each endowed with $\omega_i$ and the coalition $C'$ with $\mu_i N$ copies of agent $i$ each endowed with $x_i$.
For concision, let $y_i \coloneqq z_i + \omega_i$ for all $i \in U$ and $y'_i \coloneqq z'_i + x_i$ for all $i$.
By definition of $P(S_{k+1})$ and $V(S_{k+1})$, it follows that
\begin{align*}
    \sum_{i \in C} y_i + \sum_{i \in C'} y'_i &\leq \sum_{i \in C} \omega_i + \sum_{i \in C'} x_i \,\forall\, j \not\in S_{k+1} \\
    \sum_{i \in C} y_i + \sum_{i \in C'} y'_i &< \sum_{i \in C} \omega_i + \sum_{i \in C'} x_i \,\forall\, j \in S_{k+1}.
\end{align*}
That is, $((y_i)_{i \in C}, (y'_i)_{i \in C'})$ is a redistribution of $((\omega_i)_{i \in C}, (x_i)_{i \in C'})$.
Thus, $C$, $C'$, and $(y, y')$ constitute a rejecting coalition since $y_i \succ_i x_i$ for all $i \in U$ (since $z_i \in P_i(S_{k+1})$), $y'_i \succeq_i x_i$ for all $i \in [n]$ (since $z'_i \in Q_i(S_{k+1})$), and $\lambda_i > 0$ for some $i \in U$ (so $C \not= \varnothing$).
However, the existence of a rejecting coalition contradicts the fact $x$ lies in the rejective core.

At this point, we have shown that $P(S_{k+1}) \cap V(S_{k+1}) = \varnothing$.
Since both sets are non-empty and convex, there exists $p^{(k+1)} \in \RR^n \setminus \{0\}$ and $c \in \RR$ such that
\begin{align*}
    & p^{(k+1)} \cdot z \geq c \,\forall\, z \in P(S_{k+1}) \\
    & p^{(k+1)} \cdot z \leq c \,\forall\, z \in V(S_{k+1}).
\end{align*}
Note that $c = 0$ since $0 \in P(S_{k+1})$ (so $c \leq 0$) and $-\varepsilon \cdot \1 \in V(S_{k+1})$ for all $\varepsilon > 0$ (so $c \geq 0$).
Moreover, $p^{(k)}_j \geq 0$ for all $j$ \textit{with equality for all $j \not\in S_{k+1}$} since the elements of $V(S_{k+1})$ are free in these coordinates.
Since $p^{(k+1)} \not= 0$, it follows that $p^{(k+1)}_j > 0$ for some $j \in S_{k+1}$.

So far, we have shown that \ref{item:IH-non-negative-prices}, \ref{item:IH-simple-prices}, and \ref{item:IH-progress} hold for $k+1$.
We now turn to the rest of the inductive hypothesis, starting with \ref{item:IH-optimal-consumption}.
Note that it suffices to demonstrate the desired property only for $i \not\in T_{k+1}$ such that $p^{(k+1)} \cdot \omega_i + \alpha^{(k+1)}_i > 0$ since it already holds for all $i \in T_{k+1}$ by \ref{item:IH-optimal-consumption} for $k$.
Suppose by way of contradiction that there exists $i \not\in T_{k+1}$ and $y_i \succ_i x_i$ such that $p^{(k+1)} \cdot \omega_i + \alpha^{(k+1)}_i > 0$ yet $y_i \in C^\alpha_i(p)$.
\ref{item:IH-non-negative-prices} and the fact that $\max\{p^{(\ell)} \cdot x_i, p^{(\ell)} \cdot y_i\} \leq p^{(\ell)} \cdot \omega_i + \alpha^{(\ell)}_i = 0$ for all $\ell \in [k]$ imply that $\omega_{ij} = x_{ij} = y_{ij} = 0$ for all $j \not\in S_{k+1}$.
Thus, $y_i - \omega_i \in P_i(S_{k+1})$ and $y_i - x_i \in Q_i(S_{k+1})$, so
\[
    p^{(k+1)} \cdot y_i \geq \max\{p^{(k+1)} \cdot \omega_i, p^{(k+1)} \cdot x_i\} = p^{(k+1)} \cdot \omega_i + \alpha^{(k+1)}_i > 0
\]
Since $y_i \in C^\alpha_i(p)$ as well, the first inequality is in fact tight.
However, since $\succeq_i$ is continuous, there exists $\varepsilon > 0$ such that $(1 - \varepsilon) \cdot y_i \succ_i x_i$.
The fact that this bundle is strictly cheaper than $y_i$ contradicts the fact that $(1 - \varepsilon) \cdot y_i - \omega_i \in P_i(S_{k+1})$ and $(1 - \varepsilon) \cdot y_i - x_i \in Q_i(S_{k+1})$ as well.

A similar argument shows that \ref{item:IH-weak-cheapest-bundle} holds for $k+1$.
Consider any $i \not\in T_{k+1}$ such that $p^{(k+1)} \cdot \omega_i + \alpha^{(k+1)}_i > 0$ and any $y_i \in C^\alpha_i(p)$ such that $y_i \succeq_i x_i$.
\ref{item:IH-non-negative-prices} and the fact that $p^{(\ell)} \cdot x_i \leq p^{(\ell)} \cdot \omega_i + \alpha^{(\ell)}_i = 0$ for all $\ell \in [k]$ imply that $y_{ij} = 0$ for all $j \not\in S_{k+1}$.
Thus, $y_i - x_i \in Q_i(S_{k+1})$, so
\[
    p^{(k+1)} \cdot y_i \geq p^{(k+1)} \cdot x_i,
\]
and the weak cheapest bundle property holds for all $i \in T_{k+2} \setminus T_{k+1}$.
To conclude our analysis of \ref{item:IH-weak-cheapest-bundle}, simply note that the desired property holds for $i \in T_{k+1}$ by \ref{item:IH-weak-cheapest-bundle} for $k$.

To see why \ref{item:IH-agg-cheapest-bundle} holds for $k + 1$, suppose by way of contradiction that there exists $\beta \in \RR^n_+$ and $y \in X$ such that
\begin{align}
    & y_i \succeq_i x_i \,\forall\, i \in \supp{\beta} \notag \\
    & p^{(\ell)} \cdot \sum_i \beta_i y_i = p^{(\ell)} \cdot \sum_i \beta_i x_i \,\forall\, \ell \in [k] \label{equation:equal-in-higher-currencies} \\
    & p^{(k+1)} \cdot \sum_i \beta_i  y_i < p^{(k+1)} \cdot \sum_i \beta_i x_i \label{equation:cheaper} \\ 
    & \sum_i \beta_i y_{ij} \leq \sum_i \beta_i x_{ij} \,\forall\, j \not\in S_{k+1} \label{equation:redistribution}
\end{align}
Note that we immediately have that 
\[
    \sum_{i \in T_{k+1}} \frac{\beta_i}{\sum_{i' \in T_{k+1}} \beta_{i'}} \cdot (y_i - x_i) \in Q_{T_{k+1}}(S_{k+1})
\]
since
\begin{equation}\label{equation:redistribution-among-T_k}
    \sum_{i \in T_{k+1}} \beta_i y_{ij} \leq \sum_i \beta_i y_{ij} \leq \sum_i \beta_i x_{ij} = \sum_{i \in T_{k+1}} \beta_i x_{ij}
\end{equation}
for all $j \not\in S_{k+1}$.
The first inequality follows from the fact that $y \in \RR^{n \times n}_+$, and the second follows from Equation~\ref{equation:redistribution}.
Meanwhile, \ref{item:IH-non-negative-prices} and the fact that $p^{(\ell)} \cdot x_i \leq p^{(\ell)} \cdot \omega_i + \alpha^{(\ell)}_i = 0$ for all $\ell \in [k]$ imply that $x_{ij} = 0$ for all $i \not\in T_{k+1}$ and $j \not\in S_{k+1}$, from which the equality follows.
By definition of $p^{(k+1)}$ and membership in $Q_{T_{k+1}}(S_{k+1})$,
\begin{equation}\label{equation:T_k-not-cheaper}
    p^{(k+1)} \cdot \sum_{i \in T_{k+1}} \beta_i \cdot (y_i - x_i) \geq 0.
\end{equation}
Now, note that there cannot exist $i \in \supp{\beta} \setminus T_{k+1}$ and $j \not\in S_{k+1}$ such that $y_{ij} > 0$: otherwise, by \ref{item:IH-non-negative-prices} and Equations~\ref{equation:equal-in-higher-currencies} and~\ref{equation:redistribution-among-T_k}, $(\beta', y)$ where $\beta'_i \coloneqq \beta_i \cdot \1(i \in T_{k+1})$ for all $i \in [n]$ would yield a violation of the aggregate cheapest bundle property for some $\ell \in [k]$, contradicting \ref{item:IH-agg-cheapest-bundle} for $k$.
Thus, $y_i - x_i \in Q_i(S_{k+1})$ for all $i \in \supp{\beta} \setminus T_{k+1}$, so 
\[
    p^{(k+1)} \cdot y_i \geq p^{(k+1)} \cdot x_i \,\forall\, i \in \supp{\beta} \setminus T_{k+1}
\]
Together with Equation~\ref{equation:T_k-not-cheaper}, this inequality contradicts Equation~\ref{equation:cheaper}.
\end{proof}

\end{document}